\newif\ifsubmit     %
\newif\ifllncs      %
\newif\ifexabs      %
\newif\ifblind      %

\submittrue

\ifllncs
  \documentclass[runningheads,a4paper]{llncs}

\else
  \documentclass[letterpaper,11pt,pdfa]{article}
  \usepackage[in]{fullpage}
\fi

\usepackage{iftex}
\ifPDFTeX
  \usepackage[utf8]{inputenc}
  \usepackage[noTeX]{mmap}
  \usepackage[T1]{fontenc}
\fi
\ifLuaTeX
  \usepackage{luatex85}
  \usepackage[noTeX]{mmap}
\fi

\usepackage{mdframed}
\usepackage{amsmath}
\usepackage{amsfonts}
\usepackage{amssymb}
\usepackage{amsthm}
\usepackage{color}

\usepackage{appendix}
\usepackage{algorithm}
\usepackage{algpseudocode}
\usepackage{braket}
\usepackage{comment}
\usepackage{url}
\usepackage{bbm}
\usepackage{multicol}
\usepackage{mdframed}
\usepackage{multirow}

\usepackage[bookmarks]{hyperref}
\usepackage[nameinlink]{cleveref}
\hypersetup{
  colorlinks,%
  allcolors=blue
}

\ifllncs

  \spnewtheorem{claim}{Claim}{\bfseries}{\rmfamily}
  \crefname{claim}{claim}{claims}
  \Crefname{claim}{Claim}{Claims}
\else
  \newtheorem{theorem}{Theorem}[section]
  \newtheorem{definition}[theorem]{Definition}
  \newtheorem{remark}[theorem]{Remark}
  \newtheorem{lemma}[theorem]{Lemma}
  \newtheorem{corollary}[theorem]{Corollary}
  
  \newtheorem{example}[theorem]{Example}
  \newtheorem{claim}[theorem]{Claim}
  \newtheorem*{remark*}{Remark}

  \newtheorem*{theorem*}{Theorem}
  \newtheorem*{lemma*}{Lemma}

\newcommand{\email}[1]{\href{mailto:#1}{\texttt{#1}}}
\usepackage[style=alphabetic,minalphanames=3,maxalphanames=4,maxnames=99,backref=true]{biblatex}
  \renewcommand*{\labelalphaothers}{\textsuperscript{+}}
  \renewcommand*{\multicitedelim}{\addcomma\space}
  \DeclareFieldFormat{eprint:iacr}{Cryptology ePrint Archive: \href{https://ia.cr/#1}{\texttt{#1}}}
  \DeclareFieldFormat{eprint:iacrarchive}{Cryptology ePrint Archive: \href{https://eprint.iacr.org/archive/#1}{\texttt{#1}}}
  \AtEveryBibitem{%
    \clearlist{address}
    \clearfield{date}
    \clearfield{isbn}
    \clearfield{issn}
    \clearlist{location}
    \clearfield{month}
    \clearfield{series}
 
    \ifentrytype{book}{}{%
      \clearlist{publisher}
      \clearname{editor}
    }
  }
\fi

\usepackage{appendix}
\usepackage{algorithm}
\usepackage{algpseudocode}
\usepackage{braket}
\usepackage{comment}
\usepackage{url}
\usepackage{multicol}
\usepackage{tikz}
\usetikzlibrary{arrows,automata,positioning}
\usepackage{caption}
\usepackage[caption=false]{subfig}
\usepackage[font=small,labelfont=bf]{caption}
\usepackage{comment}
\usepackage{pifont}
\usetikzlibrary{patterns}

\usepackage{enumitem}
\setlist[description]{noitemsep}
\setlist[enumerate]{noitemsep}
\setlist[itemize]{noitemsep}

\usepackage{soul, xcolor, xparse}
\makeatletter
  \ExplSyntaxOn
    \cs_new:Npn \white_text:n #1
    {
      \fp_set:Nn \l_tmpa_fp {#1 * .01}
      \llap{\textcolor{white}{\the\SOUL@syllable}\hspace{\fp_to_decimal:N \l_tmpa_fp em}}
      \llap{\textcolor{white}{\the\SOUL@syllable}\hspace{-\fp_to_decimal:N \l_tmpa_fp em}}
    }
    \NewDocumentCommand{\whiten}{ m }
    {
      \int_step_function:nnnN {1}{1}{#1} \white_text:n
    }
  \ExplSyntaxOff
  
  \NewDocumentCommand{ \varul }{ D<>{5} O{0.2ex} O{0.1ex} +m } {%
    \begingroup
    \setul{#2}{#3}%
    \def\SOUL@uleverysyllable{%
      \setbox0=\hbox{\the\SOUL@syllable}%
      \ifdim\dp0>\z@
      \SOUL@ulunderline{\phantom{\the\SOUL@syllable}}%
      \whiten{#1}%
      \llap{%
        \the\SOUL@syllable
        \SOUL@setkern\SOUL@charkern
      }%
      \else
      \SOUL@ulunderline{%
        \the\SOUL@syllable
        \SOUL@setkern\SOUL@charkern
      }%
      \fi}%
    \ul{#4}%
    \endgroup
  }
\makeatother

\newcommand{\E}{\mathop{\mathbb{E}}}
\newcommand{\Z}{\mathbb{Z}}
\newcommand{\I}{\mathbb{I}}

\usepackage{mleftright}

\newcommand{\As}{\mathcal{A}}
\newcommand{\Bs}{\mathcal{B}}
\newcommand{\Cs}{\mathcal{C}}

\renewcommand{\P}{\sf P}

\newcommand{\prg}{{\sf PRG}}

\newcommand{\samp}{{\sf Samp}}

\newcommand{\ver}{{\sf Verify}}
\newcommand{\ch}{{\sf ch}}

\newcommand{\ans}{{\sf ans}}

\newcommand{\A}{\mathsf{A}}
\newcommand{\B}{\mathsf{B}}
\newcommand{\C}{\mathsf{C}}

\newcommand{\F}{\mathbb{F}}

\newcommand{\N}{\mathbb{N}}

\renewcommand{\pi}{{\sf PI}}

\mdfdefinestyle{figstyle}{ 
  linecolor=black!7, 
  backgroundcolor=black!7, 
  innertopmargin=10pt, 
  innerleftmargin=25pt, 
  innerrightmargin=25pt, 
  innerbottommargin=10pt 
}
\newenvironment{gamespec}{
  \begin{mdframed}[style=figstyle]}{
  \end{mdframed}}

\newcommand{\msg}{{\sf msg}}

\renewcommand{\kappa}{\ell}

\newcommand{\poly}{{\sf poly}}
\newcommand{\negl}{{\sf negl}}

\ifsubmit
    \newcommand{\luowen}[1]{}
    \newcommand{\qipeng}[1]{}
    \newcommand{\jiahui}[1]{}
    \newcommand{\markz}[1]{}
\else
    \newcommand{\luowen}[1]{{\color{magenta} Luowen: #1}}
    \newcommand{\qipeng}[1]{{\color{red} Qipeng: #1}}
    \newcommand{\jiahui}[1]{{\color{blue} Jiahui: #1}}
    \newcommand{\markz}[1]{{\color{olive} Mark: #1}}
\fi

\newcommand{\p}{\mathbf{p}}

\newcommand{\CP}{{\mathsf{CP}}}
\newcommand{\IsUniform}{{\mathsf{IsUniform}}}
\newcommand{\oneproject}{{\ket{\mathbbm{1}_{\mc{R}}}\bra{\mathbbm{1}_{\mc{R}}}}}
\newcommand{\identity}{\mathbf{I}}

\newcommand{\rR}{\mathbf{R}}
\newcommand{\rA}{\mathbf{A}}

\newcommand{\inp}{{\sf inp}}
\newcommand{\mc}[1]{{\mathcal{#1}}}
\newcommand{\mb}[1]{{\mathbf{#1}}}
\newcommand{\mbb}[1]{{\mathbb{#1}}}

\title{
Non-uniformity and Quantum Advice in the Quantum Random Oracle Model
}

\ifllncs
\author{
}
\institute{
}
\else
\author{
  Qipeng Liu\footnote{Simons Institute for the Theory of Computing. Email: \email{qipengliu0@gmail.com}}
}
\fi

\date{}

\begin{document}

\maketitle

\ifllncs 
\begin{abstract}
QROM (quantum random oracle model), introduced by Boneh et al. (Asiacrypt 2011), captures all generic algorithms. However, it fails to describe non-uniform quantum algorithms with preprocessing power, which receives a piece of bounded classical or quantum advice. 

As non-uniform algorithms are largely believed to be the right model for attackers, starting from the work by Nayebi, Aaronson, Belovs, and Trevisan (QIC 2015), a line of works investigates non-uniform security in the random oracle model. Chung, Guo, Liu, and Qian (FOCS 2020) provide a framework and establish non-uniform security for many cryptographic applications. %
Although they achieve nearly optimal bounds for many applications with classical advice, their bounds for quantum advice are far from tight. 

In this work, we continue the study on quantum advice in the QROM. We provide a new idea that generalizes the previous multi-instance framework, which we believe is more quantum-friendly and should be the quantum analog of multi-instance games. To this end, we {\emph {match}} the bounds with \emph{quantum advice} to those with \emph{classical advice} by Chung et al., showing quantum advice is almost as good/bad as classical advice for many natural security games in the QROM.

Finally, we show that for some contrived games in the QROM, quantum advice can be exponentially better than classical advice for some parameter regimes. To our best knowledge, it provides an evidence of a general separation between quantum and classical advice relative to an unstructured oracle. \end{abstract}
\else
\begin{abstract}In the quantum random oracle model (QROM) introduced by Boneh et al. (Asiacrypt 2011), a hash function is modeled as a uniformly random oracle, and a quantum algorithm can only interact with the hash function in a black-box manner. 
QRO methodology captures all generic algorithms. However, they fail to describe non-uniform quantum algorithms with preprocessing power, which receives a piece of bounded classical or quantum advice. 

As non-uniform algorithms are largely believed to be the right model for attackers, starting from the work by Nayebi, Aaronson, Belovs, and Trevisan (QIC 2015), a line of works investigates non-uniform security in the random oracle model. Chung, Guo, Liu, and Qian (FOCS 2020) provide a framework and establish non-uniform security for many cryptographic applications. %
Although they achieve nearly optimal bounds for many applications with classical advice, their bounds for quantum advice are far from tight. 

In this work, we continue the study on quantum advice in the QROM. We provide a new idea that generalizes the previous multi-instance framework, which we believe is more quantum-friendly and should be the quantum analog of multi-instance games. To this end, we {\emph {match}} the bounds with \emph{quantum advice} to those with \emph{classical advice} by Chung et al., showing quantum advice is almost as good/bad as classical advice for many natural security games in the QROM. More formally, 
\begin{itemize}
    \item {\bf OWFs:} Even with $S$-{\emph{qubits of quantum advice}}, a $T$-query quantum algorithm has advantage $O((ST+T^2)/N)$ to invert a random function with domain and range size $N$. As shown by Corrigan-Gibbs and Kogan (TCC 2019), any further improvement will lead to new classical circuit lower bounds.
    \item {\bf PRGs:} An $S$-qubit, $T$-query quantum algorithm can distinguish between a random image and a random element in the range, with an winning probability at most $1/2+O(T^2/N)^{1/2}+O(ST/N)^{1/3}$, in contrast to $1/2+O((S^5T+S^4T^2)/N)^{1/19}$ by Chung et al.
    \item {\bf Salting:} A commonly used mechanism in cryptography called salting defeats preprocessing, even with quantum advice, improved the bounds by Chung et al.  
\end{itemize}
Finally, we show that for some contrived games in the QROM, quantum advice can be exponentially better than classical advice for some parameter regimes. To our best knowledge, it provides the first evidence of a general separation between quantum and classical advice relative to an unstructured oracle. 

\end{abstract}
\fi

\section{Introduction}

Many practical cryptographic constructions are analyzed in idealized models, for example, the random oracle model which treats an underlying hash function as a uniformly random oracle (ROM) \cite{bellare1993random}. On a high level, the random oracle model captures all algorithms that use the underlying hash function in a generic (black-box) way; often, the best attacks are generic. 
Whereas the random oracle methodology guides the actual security of practical constructions, it fails to describe non-uniform security: that is, an algorithm consists of two parts, the offline and the online part; the offline part can take forever, and at the end of the day, it produces a piece of bounded advice for its online part; the online part given the advice, tries to attack cryptographic constructions efficiently.

Non-uniform algorithms are largely believed to be the right model for attackers and usually show advantages over uniform algorithms \cite{C:Unruh07,coretti2018random,coretti2018non}. The famous non-uniform example is Hellman's algorithm \cite{hellman1980cryptanalytic} for inverting permutations or functions. When a permutation of range and domain size $N$ is given, Hellman's algorithm can invert any image (with certainty) with roughly advice size $\sqrt{N}$ and running time $\sqrt{N}$. In contrast, uniform algorithms require running time $N$ to achieve constant success probability. Another more straightforward example is collision resistance. When non-uniform algorithms are presented, no single fixed hash function is collision-resistant as an algorithm can hardcode a pair of collisions in its advice. %

Non-uniform security in idealized models has been studied extensively in the literature. Let us take the two most simple yet fundamental security games as examples: one search game and one decision game. The first one is one-way function inversion (or OWFs) as mentioned above. The goal is to invert a random image of the random oracle. The study was initialized by Yao \cite{yao1990coherent} and later improved by a line of works \cite{de2010time,C:Unruh07,dodis2017fixing,coretti2018random}. They show that any $T$-query algorithm with arbitrary $S$-bit advice, can win this game with probability at most $\tilde{O}(ST/N)$, assuming the random oracle has equal domain and range size. The other example is pseudorandom generators (or PRG). The task is to distinguish between a random image $H(x)$ ($x$ is uniformly at random and $H$ is the hash function) or a random element $y$ in its range. Since it is a decision game, some techniques for OWFs may not apply to PRGs, which we will see later. Its non-uniform security is $O(1/2+T/N + \sqrt{ST/N})$ by
Coretti et al. \cite{coretti2018random}, and later improved 
by Garvin et al. \cite{gravin2021concentration}. 

The quantum setting is very similar to the classical one, except an algorithm can query the random oracle in superposition. Boneh et al. \cite{AC:BDFLSZ11} justify the ability to make superposition queries since a quantum computer can always learn the description of a hash function and compute it coherently. Besides, advice can be either a sequence of {\bf bits} or {\bf qubits}. We should carefully distinguish between the two different models. Indeed, we believe non-uniform quantum algorithms with quantum advice are important to understand and should be considered the ``right'' attacker model when full-scale quantum computers are widely viable and quantum memory is affordable. 

Nayebi, Aaronson, Belovs, and Trevisan \cite{nayebi2014quantum} initiated the study of quantum non-uniform security with classical advice of OWFs and PRGs. Hhan, Xagawa and Yamakawa \cite{hxy19}, Chung, Liao and Qian \cite{chung2019lower} extended the study to quantum advice. Most recently, Chung, Guo, Liu and Qian \cite{chung2020tight} improved the bounds for both examples. For OWFs, their bounds are almost optimal in terms of query complexity for both classical and quantum advice. They show that to invert a random image with at least constant probability, advice size $S$ and the number of queries $T$ should satisfy $S T + T^2 \geq \tilde{\Omega}(N)$. However, a gap between classical and quantum advice appears when we choose security parameters for practical hash functions against non-uniform attacks. In practice, we ensure that an adversary with bounded resources (for example, $S = T = 2^{128}$) only has probability smaller than $2^{-128}$. The bounds in \cite{chung2020tight} suggest that for OWF, the security parameter needs to be $n = 384$ (and $N = 2^{384}$) for classical advice and $n = 640$ for quantum advice, leaving a big gap between two types of advice. Even worse, when it comes to PRGs, the security parameters are $n = 640$ for classical advice v.s. $n = 3200$ for quantum advice; not to mention a large gap between their query complexity, unlike OWFs. 

As understanding quantum advice is beneficial to both practical cryptography efficiency and may inspire general computation theory (such as, ${\sf QMA}$ v.s. ${\sf QCMA}$ \cite{aaronson2007quantum,aaronson2021open} and ${\sf BQP/poly}$ v.s. ${\sf BQP/qpoly}$ \cite{aaronson2004limitations}), we raise the following natural question: 

\begin{quote}
{\it Can quantum advice outperform classical advice in the QROM?}
\end{quote}

In this work, we provide a new technique for analyzing quantum advice in the QROM and show that for many games, the non-uniform security with quantum advice matches the best-known security with classical advice, including OWFs and PRGs. It gives strong evidence that for many cryptographic games in the QROM, quantum advice provides no or little advantage over classical one. 

So far, we have seen no advantage of quantum advice in the QROM for common cryptographic games. We then ask the second question: 

\begin{quote}
{\it Is there any (contrived) game in the QROM, in which quantum advice is ``exponentially better'' than classical advice? }
\end{quote}
We give an affirmative answer to this question, for some parameters of $S, T$. We show that when algorithms can not make online queries (i.e., $T = 0$), there is an exponential separation between quantum and classical advice for certain games. This result is inspired by the recent work by Yamakawa and Zhandry \cite{yamakawa2022verifiable} on verifiable quantum advantages in the QROM. We elaborate on both results now.

\subsection{Our Results}

Our first result is to give a quantum analog of ``multi-instance games'' via ``alternating measurement games'' (introduced in \Cref{sec:main_theorem}) and develop a new technique for analyzing non-uniform bounds with quantum advice.
Our techniques do not need to rewind a non-uniform quantum algorithm and completely avoid the rewinding issues/difficulties in the prior work \cite{chung2020tight}. We delay the technical details in \Cref{sec:overview} and give other results below. 

To show the power of our technique, we incorporate it into three important applications: OWFs, PRGs, and salted cryptography. 
Note that our result below is a non-exhaustive list of applications. With little effort, we can show improved non-uniform security with quantum advice of Merkle-Damg\r ard \cite{guo2021unifying}, Yao's box \cite{chung2020tight} and other games. 

\paragraph{One-Way Functions.} 

In this application, a random oracle is interpreted as a one-way function. A (non-uniform) algorithm needs to win the OWF security game with the random oracle as a OWF. Formally, let $H: [N] \to [M]$ be a random oracle. 
\begin{enumerate}
    \item A challenger samples a uniformly random input $x \in [N]$ and sends $y = H(x)$ to the algorithm.
    \item The algorithm returns $x'$ and it wins if and only if $H(x') = y$. 
\end{enumerate}

When both advice and queries are classical, the best lower bound is $\tilde{O}(ST/\alpha)$ by \cite{coretti2018random},  where $\alpha = \min\{N, M\}$ and $N$, $M$ are the domain and range size of the random oracle. In other words, no algorithm with $S$ bits of advice and $T$ classical queries can win with probability more than $\tilde{O}(ST/\alpha)$. There is a gap between this lower bound and the upper bound $\approx T/\alpha + (S^2 T/\alpha^2)^{1/3}$ provided by Hellman's algorithm\footnote{Hellman's algorithm on functions does not behave as well as on permutations. Upper and lower bounds meet at $ST/\alpha$ only when we consider permutations.}. Later, Corrigan-Gibbs and Kogan \cite{corrigan2019function} study the possible improvement on the lower bound and conclude that any improvement will lead to improved results in circuit lower bounds. Thus, $\tilde{O}(ST/\alpha)$ is the best one can hope for in light of the barrier. 

Chung et al. \cite{chung2020tight} show that if $S$ bits of classical advice and $T$ quantum queries are given, the maximum winning probability is bounded by $\tilde{O}\left( \frac{S T + T^2}{\alpha} \right)$.
They further argue that this bound is almost optimal. Intuitively, one can think of this as $T^2/\alpha$ comes from a brute-force Grover's algorithm \cite{grover1996fast}, without using any advice, and $ST/\alpha$ comes from classical advice and hits the classical barrier by \cite{corrigan2019function}. 

For quantum advice and quantum queries, they show the maximum success probability is $\tilde{O}\left( \frac{S T + T^2}{\alpha} \right)^{1/3}$. As mentioned early, although the bound is optimal regarding query complexity, the exponent seems non-tight. Thus, they ask the following question:
\begin{quote}
\it ... Can this loss (of the exponent) be avoided, or is there any speed up in terms of $S$ and $T$ for sub-constant success probability?. 
\end{quote}

Our first result gives a positive answer to the above question and proves that the loss on exponent can be avoided. 
\begin{theorem}\label{thm:simple_OWF}
Let $H$ be a random oracle $[N] \to [M]$ and $\alpha = \min\{N, M\}$. 
One-way function games in the QROM have security $O\left( \frac{ST + T^2}{\alpha}\right)$ against non-uniform quantum algorithms with $S$-qubits of advice and $T$ quantum queries. 
\end{theorem}

The theorem guides security parameter choices of hash functions to be secure against non-uniform attacks. The security parameter $n$ should be $384$ to have security $2^{-128}$ against non-uniform quantum attacks with $S = T = 2^{128}$. 
Another direct implication of our theorem is that, when quantum advice $S = O(\sqrt{\alpha})$, quantum advice is useless for speeding up function inversion. To put it in another way, Grover's algorithm can not be sped up and only has probability $T^2/\alpha$ to succeed even with quantum advice of size $O(\sqrt{\alpha})$, relative to a random oracle. We list a comparison of best-known bounds and our result below. 

\begin{table}[!hbt]
\begin{center}
\begin{tabular}{| c | c | c |}
\hline
 Classical Advice in \cite{chung2020tight} & Quantum Advice in \cite{chung2020tight} &  Quantum Advice in This Work \\
 $\tilde{O}\left( \frac{S T + T^2}{\alpha} \right)$ & $\tilde{O}\left( \frac{S T + T^2}{\alpha} \right)^{1/3}$  & ${O}\left( \frac{S T + T^2}{\alpha} \right)$ \\  
\hline
\end{tabular}
\caption{Non-uniform security for OWFs with $T$ queries and $S$ bits (qubits) of advice, where $\alpha = \min\{N, M\}$ and $N$, $M$ are the domain and range size of the random oracle. Our bound is a ``big-$O$'' instead of ``big-$\tilde{O}$'' as we also remove the dependence on $\log N$ and $\log M$.}
\label{fig:OWF_result}
\end{center}
\end{table}

\paragraph{Pseudorandom Generators.}

Another important application we will focus on is pseudorandom generators.
One fundamental difference from one-way functions is its being a decision game. We will later see that publicly verifiable games such as one-way functions are easy to deal with in the previous work \cite{chung2020tight}. For games that can not be publicly verified, such as decision games, \cite{chung2020tight} often gives worse bounds. 

In this game, an algorithm tries to distinguish between an image of a random input, and a uniformly random element in the range. Let $H: [N] \to [M]$ be a random oracle.  
\begin{itemize}
    \item A challenger samples a uniformly random bit $b$. If $b = 0$, it samples a uniformly random $x \in [N]$ and outputs $y = H(x)$; otherwise, it samples a uniform $y \in [M]$ and outputs $y$. 
    \item The algorithm is given $y$ and returns $b'$. It wins if and only if $b' = b$. 
\end{itemize}

Our new technique demonstrates the following theorem about PRGs. 
\begin{theorem}
Let $H$ be a random oracle $[N] \to [M]$. 
PRG games in the QROM have security $1/2 +  O\left( \frac{T^2}{N}\right)^{1/2} + O\left( \frac{ST}{N}\right)^{1/3}$ against non-uniform quantum algorithms with $S$-qubits of advice and $T$ quantum queries. 
\end{theorem}

\begin{table}[!hbt]
\begin{center}
\begin{tabular}{| c | c | c |}
\hline
 Classical Advice in \cite{chung2020tight} & Quantum Advice in \cite{chung2020tight} &  Quantum Advice in This Work \\
 $\frac{1}{2} + \tilde{O}\left( \frac{S T + T^2}{N} \right)^{1/3}$ & $\frac{1}{2} + \tilde{O}\left( \frac{S^5 T + S^4 T^2}{N} \right)^{1/19}$  & $\frac{1}{2} + O\left(\frac{T^2}{N}\right)^{1/2} + O\left( \frac{ST}{N}\right)^{1/3}$ \\  
\hline
\end{tabular}
\caption{Non-uniform security of PRGs with $T$ queries and $S$ bits (qubits) of advice. Our bound also improves the previous result on classical advice by reducing the exponent on $T^2/N$ from $1/3$ to $1/2$; we note that the improvement on the exponent only follows from a simple observation and can also be applied to the previous work as well.}
\label{fig:PRG_result}
\end{center}
\end{table}

\paragraph{``Salting Defeats Preprocessing''.}
Finally, instead of proving more concrete non-uniform bounds like Merkle-Damg\r ard \cite{guo2021unifying}, we demonstrate that the generic mechanism ``salting'' helps prevent quantum preprocessing attacks even with quantum advice. Maybe the most illustrating example is collision-resistant hash functions. As mentioned before, no single fixed hash function can be collision resistant against non-uniform attacks. A typical solution is to add ``salt'' to the hash function. A salt is a piece of random data that will be fed into a hash function as an additional input. To attack a salted collision resistant hash function, an adversary gets a salt $s$ and is required to come out with two input $m \ne m'$ such that the hash evaluation on $(s, m)$ equals that of $(s, m')$. Intuitively, since salt $s$ is chosen uniformly at random from a large space, advice is not long enough to include collisions for every possible salt. 
Thus, salting is a mechanism that compiles a game into another game, by adding a random extra input $s$ and restricting the execution of the game always under oracle access to $H(s,\cdot)$. 

Chung et al. \cite{chung2013power}, and Coretti et al. \cite{coretti2018random} formally proved the non-uniform security of salted collision-resistant hash in the classical ROM. Chung et al. \cite{chung2020tight} extended the statement in the quantum setting. For quantum advice, their result roughly says that if an underlying game $G$ is publicly verifiable or a decision game, then the salted version of $G$ is secure against non-uniform attacks.

Our third results improve the prior ones in two different aspects. First, our theorem works not only for publicly verifiable or decision games, but for any types of games (see our definition of games \Cref{def:game}). Second, our theorem is tighter and provides a more pictorial statement for ``salting defeats preprocessing'', elaborated below. Our bounds match those with classical advice in \cite{chung2020tight}. 

\begin{theorem}[Informal, \Cref{prop:salt}]\label{thm:salt_defeats}
For any game $G$ in the QROM, let $\nu(T)$ be its uniform security in the QROM. 
Let $G_S$ be the salted game with salt space $[K]$. Then $G_S$ has security $\delta(S, T)$ against non-uniform quantum adversaries with $T$ queries and $S$-qubits of advice,
\begin{enumerate}
    \item $\delta(S, T) \leq 4 \nu(T) + O(S T/K)$;
    \item If $G_S$ is a decision game, then $\delta(S, T) \leq \nu(T) + O({S T/K})^{1/3}$. 
\end{enumerate}
\end{theorem} 

That is to say, the non-uniform security of $G_S$ and uniform security of $G$ only differs by a term of $O(ST/K)$ or $O(ST/K)^{1/3}$ depending on the type of the game. 
When the game $G$ is a search game, $G_S$ has non-uniform security $4 \nu(T) + O(S T/K)$. We can choose $S$ to ensure $ST/K \leq \nu(T)$ so that the non-uniform security of $G_S$ is in the same order of $G$'s security $\nu(T)$. For decision games, we choose $S$ such that $(ST/K)^{1/3}$ is extremely small.

In \cite{chung2020tight}, they show that for publicly verifiable games, $\delta := \delta(S, T)$ satisfies 
$\delta \leq \tilde{O}\left(\nu(T/\delta)  + \frac{ST}{K\delta}\right)$ whereas ours works for any games and $\delta(S, T) \leq 4 \nu(T) + O(ST/K)$. For decision games, ours also significantly improves prior results (see \Cref{fig:salting_fig} and Theorem 7.6 in \cite{chung2020tight} for a comparison). 
The dependence in their theorems on uniform security $\nu$ is much more complicated and yields loose bounds. 
Most notably, for decision games, when the salt size  $K\to \infty$, the bound in \cite{chung2020tight} does not rule out the speed up from having $S$-qubits of advice (corresponding to the term $\nu'(S^2T/\epsilon^8)$); whereas our bound gives $\nu(T)$ --- exactly the security in the uniform case, completely ruling out the influence of quantum advice. 

\begin{table}[!hbt]
\begin{center}
\begin{tabular}{| c | c | c |}
\hline
  & Quantum Advice in \cite{chung2020tight} &  Quantum Advice in This Work \\
 Any Games  & $\delta \leq \tilde{O}\left( \nu(T/\delta) + ST/(K \delta) \right)$ & $\delta \leq 4 \nu(T) + O(ST/K)$  \\  
 & & \\
 \multirow{2}{*}{Decision Games} &  $\delta \leq 1/2 + \epsilon$   &   \multirow{2}{*}{$\delta \leq \nu(T) + O(ST/K)^{1/3}$}   \\  
   & {\small where $\epsilon \leq \tilde{O}\left( \nu'(S^2 T/\epsilon^8) + \sqrt{S^5 T/(K \epsilon^{17})} \right)$}  &   \\  
   & {\small and $\nu'(T) := \nu(T) - 1/2$}  &   \\  
\hline
\end{tabular}
\caption{Salting ``defeats'' preprocessing.}
\label{fig:salting_fig}
\end{center}
\end{table}

\paragraph{Separation of Quantum and Classical Advice in the QROM.}

So far, we have seen many examples that quantum advice is as good/bad as classical advice. Below, we show that it is not always the case in the QROM: there exists a game in the QROM such that quantum advice is exponentially better than classical advice. 

\begin{theorem}[Separation of Quantum and Classical Advice in the QROM] \label{thm:informal_advice}
Let $H$ be a random oracle $[2^{\poly(n)}] \to \{0,1\}$.
There exists a game $G$ in the QROM such that, 
\begin{itemize}
    \item $G$ has security $2^{-\Omega(n)}$ against non-uniform adversaries with $S$-bits of {\bf classical} advice and making no queries, for $S = 2^{n^c}/n$ and some constant $0 < c < 1$; 
    \item There is a non-uniform adversary with $S$-qubits of {\bf quantum} advice and making no queries, that achieves winning probability $1 - \negl(n)$, for $S = \tilde{O}(n)$. 
\end{itemize}
\end{theorem}

Although the bound only works in the parameter regime $T = 0$, to our best knowledge, it is the first example of an exponential separation between quantum and classical advice in the QROM (or for inputs without structures). 

\begin{remark}
For the parameter regime $T = 0$, the above separation can be alternatively viewed as an exponential separation of quantum/classical one-way communication complexity for some relation $\mc{R} \subseteq \mc{X} \times \mc{Y} \times {Z}$. In the context of one-way communication complexity, there are two players, Alice and Bob. Alice gets an input $x \in \mc{X}$ and Bob gets an input $y \in \mc{Y}$; Alice sends one (classical or quantum) message to Bob and Bob tries to output $z \in \mc{Z}$ such that $(x, y, z) \in \mc{R}$. Our result in \Cref{thm:informal_advice} is a separation of quantum/classical one-way communication complexity when $\mc{X} = \{0,1\}^{2^{\poly(n)}}$, $\mc{Y} = \{0,1\}^n$, $\mc{Z} = \{0,1\}^{n \times \poly(n)}$; when the message is allow to be quantum, $\tilde{O}(n)$ qubits are sufficient; on the other hand, the classical communication complexity is $\Omega(2^{n^c}/n)$. 

Exponential separation of quantum/classical one-way communication complexity is already known, starting from the work by \cite{bar2004exponential} (later by \cite{gavinsky2008classical}) based on the so-called hidden matching problem. We believe the hidden matching problem can be also turned into a separation  of quantum/classical advice in the parameter regime $T = 0$, in the QROM. However, \cite{bar2004exponential} only proved \emph{average-case} hardness against \emph{deterministic} classical Bob. Therefore, we pick the recent result by Zhandry and Yamakawa for simplicity of presentation. 
\end{remark}

\subsection{Organization}
The rest of the paper is organized as follows. In \Cref{sec:overview}, we give an overview of our main technical contribution and achieve non-uniform bounds for OWFs. \Cref{sec:prelim} and \Cref{sec:gamesandalgorithms} recall the notations and backgrounds on quantum computing, random oracles models, non-uniform security and bit-fixing models. 
\Cref{sec:notations} introduces decomposition of advice with respect to a game, which helps the proof of our main theorem. 
\Cref{sec:main_theorem} proves the main theorem whereas \Cref{sec:application} applies the main theorem to various applications. Finally in \Cref{sec:separation}, we give the separation of quantum and classical advice.

\ifblind
\else
\section*{Acknowledgements}
We would like to thank Kai-Min Chung for his discussion on an early write-up and providing an intuitive explanation of the decomposition of quantum advice in our work; Jiahui Liu and Luowen Qian for their comments on an early draft of this paper; Luowen Qian and Makrand Sinha for mentioning the connections between our impossibility result and quantum one-way communication complexity. 

Qipeng Liu is supported in part by the Simons Institute for the Theory of Computing, through a Quantum Postdoctoral Fellowship, by the DARPA SIEVE-VESPA grant Np.HR00112020023 and by the NSF QLCI program through grant number OMA-2016245. Any opinions, findings and conclusions or
recommendations expressed in this material are those of the author(s) and do not necessarily reflect the views of the United States Government or DARPA.
\fi

\section{Technical Overview}
\label{sec:overview}

This overview will primarily focus on OWF games for the random oracle $H$ with the same domain and range. We will turn to PRGs when we discuss the difficulty of decision games compared to search games. The same ideas in OWF games will apply to other applications as well. 

\paragraph{Recap \cite{chung2020tight} for Classical Advice.}

We start by recalling the ideas for classical advice behind \cite{chung2020tight}. Let $\As$ be any $T$-query non-uniform algorithm with $S$-bits of classical advice for OWF games. For convenience, we call such algorithm $(S, T)$ algorithm with classical advice. 
Inspired from \cite{aaronson2004limitations}, \cite{chung2020tight} shows that if $A$ has $\delta$ success probability in winning the OWF game, then one can run $\As$ multiple times and win the following sequential $g$-multi-instance version\footnote{The case of $g$-instances being given in parallel was consider in \cite{aaronson2004limitations}. \cite{chung2020tight} improved over the idea and proposed the sequential version, which was shown to give better implications comparing to the parallel case.} of OWF games with probability roughly $\delta^g$: 
\begin{figure}[!hbt]
    \centering
    \begin{gamespec}
    \begin{itemize}
        \item For each round $i \in [g]$, a challenger samples a random image and gives it to an algorithm.
        \item The algorithm has $T$ queries in the $i$-the round and outputs an alleged preimage for the $i$-th image.
        \item The algorithm wins if and only if it is correct in all the rounds. 
    \end{itemize}
    \end{gamespec}
    \caption{Multi-Instance Games for OWFs.}
    \label{fig:mis-classical-advice}
\end{figure}

Since $\As$ has only classical advice, one can always reset the whole algorithm and start $\As$ from scratch for each round. It is easy to observe that running and rewinding $\As$ for each stage achieves advantage (winning probability) $\delta^g$. This reduction (step 1 in \Cref{fig:reduction_OWF}) is the main challenge for quantum advice, as resetting and rewinding a non-uniform quantum algorithm is generally very difficult. We will discuss it in the next section. In the last step, we can completely remove advice by replacing the advice with a random guess (step 2 in \Cref{fig:reduction_OWF}) and introduce a multiplicative loss $2^{-S}$. As a consequence, we obtain a uniform algorithm for $g$-multi-instance games with advantage $2^{-S} \delta^g$ from $\As$. 

Therefore, to upper bound the success probability $\delta$ for OWF games, we investigate the maximal advantage $\varepsilon^g$ of uniform algorithms in the $g$-multi-instance games for $g := S$. Clearly, $\delta \leq 2 \varepsilon$ from $2^{-S} \delta^g \leq \varepsilon^g$ for $g = S$: as there exists a uniform algorithm with winning probability $2^{-S} \delta^g$, but the chance can not be greater than $\varepsilon^g$. 
\cite{chung2020tight} show that for OWFs, the advantage of algorithms with {\bf classical} or {\bf quantum} advice and $T$ queries in each round is bounded by $\varepsilon^g$ for $\varepsilon \approx (ST + T^2)/N$. Therefore, $\delta$ is ${O}((ST+T^2)/N)$, concluding the proof of the main theorem in their work. We demonstrate the idea in \Cref{fig:reduction_OWF}.

We omit many details in the above discussion --- most notably, the analysis of uniform security in the sequential multi-instance games (step 3). The discussion is delayed to the end of this section, when it is needed. 

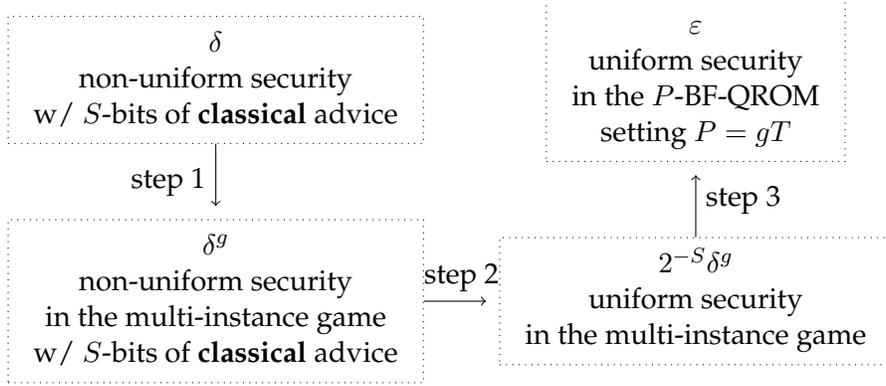
\begin{figure}[!hbt]
\centering
			\begin{tikzpicture}[scale=1,->,shorten >=5pt]
			    \node (first) [rectangle, dotted, draw=black] {\begin{tabular}{c} $\delta$ \\  non-uniform security \\ w/ $S$-bits of {\bf classical} advice \end{tabular}};
			    \node (second) [rectangle, dotted, draw=black, below=of first] {\begin{tabular}{c} $\delta^g$ \\  non-uniform security \\ in the multi-instance game \\ w/ $S$-bits of {\bf classical} advice \end{tabular}};
			    \node (third) [rectangle, dotted, draw=black, right=of second]{\begin{tabular}{c} $2^{-S} \delta^g$ \\ uniform security \\ in the multi-instance game \end{tabular}};
			    \draw[->] (first) -- (second) node [midway, left] {step 1};
			    \draw[->] (second) -- (third) node [midway, above] {step 2};
			    \node (fourth) [rectangle, dotted, draw=black, above=of third] {\begin{tabular}{c} $\varepsilon$ \\ uniform security \\ in the $P$-BF-QROM \\ setting $P = gT$ \end{tabular}};
			    \draw[->] (third) -- (fourth) node [midway, right] {step 3};
			\end{tikzpicture}
        \caption{From non-uniform security with classical advice to multi-instance security for OWFs, in \cite{chung2020tight}.}
        \label{fig:reduction_OWF}
\end{figure}

When it comes to PRGs, setting $g = S$ no longer works. The $2$ factor in $\delta \leq 2 \epsilon$ leads to a trivial bound because $\varepsilon$ is about $1/2$ for decision games. By appropriately choosing $g = S/\gamma$ for some $\gamma \in (0,1)$, the exact same idea applies. %

\paragraph{Difficulties of Rewinding Quantum Advice in \cite{chung2020tight}.}

The reduction in \Cref{fig:reduction_OWF} is designated to classical advice. When advice is quantum, step 2 still works, but step 1 does not anymore. Guessing a quantum advice of size $S$ only introduces a multiplicative loss by at most $2^{-S}$, following \cite{aaronson2004limitations}. However, step 1 requires rewinding the non-uniform quantum algorithm $\As$ with quantum advice $g$ times. Since we will eventually set $g = S$ and $S$ can be arbitrarily large, there is no guarantee for $g$ consecutive rewindings. Even worse, as the success probability $\delta$ of $\As$ can be very small, a single rewinding may not even be possible. 

The solution in \cite{chung2020tight} is to boost the success probability of $\As$ to almost $1$, using multiple copies of the same advice. When the probability is close to $1$, one can gently measure outcomes in each round and rewind a non-uniform algorithm for $g$ consecutive times\footnote{There is a missing caveat. The correct answer can be non-unique, as in OWF games. There is an easy fix for this issue in \cite{chung2020tight}. We simply ignore it and assume answers are unique, as we do not need the fix and it does not change the main idea. }. 
Assume there are $k$ copies of the oracle-dependent quantum advice $\ket{\sigma_H}$. To invert an image $y$, an algorithm $\Bs$ runs $k$ copies of $\As$ on input $y$ with advice $\ket{\sigma_H}$ in parallel; each $\As$ produces an alleged pre-image $x_i$; $\Bs$ verifies and outputs the right answer by checking whether $H(x_i) = y$. Since each instance of $\As$ wins with probability $\delta$, appropriately choosing $k$ will close the success probability to $1$ and allow $g$ consecutive rewindings. 

However, the above solution gives the lower bound $O((ST+T^2)/N)^{1/3}$ for OWFs compared to its classical advice counterpart $O((ST+T^2)/N)$. This is due to the need for rewinding and multiple copies of quantum advice. 

\medskip

The looseness in the exponents in OWFs is not the worst. For PRGs, the approach above does not work at all. While in OWF games, $\Bs$ can pick the correct answer as long as it exists by checking whether $H(x_i) = y$ for all $i$; it is not the case in PRG games. The reduction algorithm $\Bs$ has no way to tell if an answer is correct, since PRG games are not publicly verifiable. Another attempt is to let $\Bs$ do a majority vote over the outcomes from many copies of $\As$. \cite{chung2020tight} show that this approach does not behave as expected: even if a non-uniform algorithm can answer correctly w.p. $60\%$, the majority vote can pull the chance down to $40\%$, not $99\%$, even worse than a random guess!

Their answer is to ``gently'' estimate the success probability in each round of the multi-instance game and flip a coin according to this estimated probability. By utilizing an online version of shadow tomography \cite{aaronson2019gentle}, they achieve lower bounds for PRG games with quantum advice. However, the bounds for quantum advice are $1/2+O((S^5 T+ S^4 T^2)/N)^{1/19}$ compared to the bounds $1/2+O((ST+T^2)/N)^{1/3}$ for classical advice.

\medskip

One may also try to apply other rewinding techniques to step 1, for example, the ``measure-and-repair'' approach by Chiesa et al. \cite{chiesa2022post}. The tool introduces an inverse polynomial loss on success probability for every rewinding and requires the valid outcomes of the game to satisfy some form of collapsing properties. Thus, it is unlikely that their technique can be applied to this setting, as collapsing does not hold for general games, and exponentially many rewindings result in a huge loss. With the aforementioned difficulties, we start thinking about if multi-instance games (\Cref{fig:mis-classical-advice}) are the right way to go?

\paragraph{Quantum Advice as Maximizing Overlaps (\Cref{sec:notations}).}
For any non-uniform quantum algorithm with quantum advice for OWFs, $\As$ can be written formally in two parts:
\begin{enumerate}
    \item a non-uniform oracle-dependent advice $\{\ket{\sigma_H}\}_H$, and
    \item a uniform algorithm (unitary) $\{U_y\}_{y \in [N]}$. 
\end{enumerate} On input a challenge $y$ and oracle access to $H$, it operates as follows: prepares $\ket{\sigma_H}\ket{0^L}$ and applies $U^H_y$ on its internal register; measures and outputs the first $n$ bit of the registers as an answer. 
Since $\ket{\sigma_H}$ is only of $S$-qubits, the rest of the input should be independent of $H$, and we thereby model it as $\ket{0^L}$ for any $L$ (it can even be exponentially large, as we only care about queries not running time in the QROM).
The verification procedure can be written as a projector $V^H_x$ on the registers, and output $1$ if and only if $H(x) = H(x')$ assuming the first $n$ bit in the computational basis is $x'$. Due to the operational meaning of $U$ and $V$, the success probability when given oracle access to $H$ can be then written as
\begin{align*}
    \delta_H & = \mathbb{E}_{x}\left[ \left|  V^H_x U^H_{H(x)} \ket{\sigma_H} \ket{0^L} \right|^2 \right].
\end{align*}
The above probability describes the progress of sampling a random challenge $x$, feeding $H(x)$ as input to the non-uniform $\As$ and checking whether $\As$'s answer is correct with respect to $x$.

Here is an alternative way to look at $\delta_H$. 
Define $P^H$ as the following Hermitian matrix: $P^H = \mathbb{E}_x\left[(U^H_{H(x)})^\dagger V^H_x U^H_{H(x)} \right]$. $\delta_H$ can be alternatively written in terms of $P^H$ and the starting state: 
\begin{align*}
\delta_H = \langle \sigma_H, 0^L |P^H |\sigma_H, 0^L\rangle.
\end{align*}
As $P^H$ is Hermitian and ${\bf 0} \preceq P^H \preceq {\bf I}$, $P^H$ has an eigen-decomposition with real eigenvalues in $[0, 1]$. Without loss of generality, we assume the eigenvectors $\ket{\phi_p}$ have distinct eigenvalues $p \in [0,1]$ and thus $P^H = \sum_p p \ket{\phi_p} \bra{\phi_p}$.
Each $\ket{\phi_p}$ together with $\{U_y\}_y$ is an quantum algorithm whose success probability in the OWF game equals to $p$, as $\langle \phi_p | P^H | \phi_p\rangle = p \langle \phi_p | \phi_p\rangle = p$. 

Then the success probability $\delta_H$ can be written in terms of eigenvalues, eigenvectors of $P^H$ and the projection of $\ket{\sigma_H, 0^L}$ under the eigenbasis: 
\begin{align} \label{eq:decomposition}
    \delta_H = \sum_p |\alpha_p|^2 p \quad\text{ where }\ket{\sigma_H, 0^L} = \sum_p \alpha_p \ket{\phi_p}.
\end{align}
A nature analogy of \Cref{eq:decomposition} in the classical setting is that a (randomized) algorithm can be decomposed into a collection of other algorithms, each is picked with certain probability; the probability of the larger algorithm will be a convex combination of those smaller algorithms. In the quantum case, as shown in \Cref{eq:decomposition}, the decomposition is still possible but we need to work under the eigenbasis of $P^H$: %
the non-uniform quantum algorithm is in superposition of $\ket{\phi_p}$ (a quantum algorithm with winning probability $p$) with amplitude being $\alpha_p$.

A further interpretation of \Cref{eq:decomposition}  tells us that to maximize $\delta_H$, one needs to pick an appropriate $\ket{\sigma_H}$ such that the overlap between $\ket{\sigma_H}\ket{0^L}$ and eigenvectors with large eigenvalues is as large as possible. One extreme example is when $\ket{\sigma_H}$ has unbounded length; in this case, we can always set $\ket{\sigma_H} := \ket{\phi_{p^*}}$ for the largest $p^*$ and $L = 0$, in which a success probability $p^*$ is achieved. However, this is not always possible as $\ket{\sigma_H}$ has only $S$-qubits and prevents us from maximizing the overlap. 

Because our target $\delta = \mathbb{E}_H[\delta_H]$, the first attempt to bound $\delta$ is to look at eigenvalues $\{p\}$ and distributions (amplitudes) of $\{p\}$ for each $H$ individually. This approach is very difficult to analyze as the structure of $H$ significantly affects both $\{p\}$ and its distribution (amplitudes). For example, when $H$ is an all-zero function, the largest eigenvalue is $1$; since an algorithm always wins when all the images are $0^n$. Whereas for an overwhelming fraction of $H$, it is far away from $1$. We do not know how to analyze $\delta_H$ individually; if possible, it must be laborious. 

\paragraph{Step 1: Alternating Measurement Games (\Cref{sec:main_theorem}).}
This is the analog of step 1 in \Cref{fig:reduction_OWF}. 
Let $\p$ be the random variable of the eigenvalues (probability) when $\ket{\sigma_H, 0^L}$ is projected into the eigenbasis of $P^H$. 
Recall that $\delta = \mathbb{E}_H[\sum_p |\alpha_p|^2 p] = \mathbb{E}_{H}[\p]$.
Our core idea is to give a global characterization of the distribution of $\p$. Namely, we want to bound the $g$-th moment of the random variable $\p$ for some large $g$: $\mathbb{E}_{H}[\p^g]$. If the $g$-th moment is $\varepsilon^g$, that means $\p$ concentrates round $\varepsilon$. 
More formally, by Jensen's inequality, 
\begin{align*}
    \mathbb{E}_H[\p] \leq \left(\mathbb{E}_H{[\p^g]}\right)^{1/g}, \text{ for all } g \geq 1. 
\end{align*}

If we can find a game whose success probability is $\mathbb{E}_H[\p^g]$ and upper bound the probability, we can also upper bound $\mathbb{E}_H[\p]$. Inspired by the alternating measurement technique by Marriot and Watrous \cite{marriott2005quantum}, we come up with the following games, which we call ``alternating measurement games'' and show the success probability in this game is precisely $\mathbb{E}_H[\p^g]$. For those who are familiar with \cite{marriott2005quantum} and its applications in cryptography (\cite[\dots]{zhandry2020schrodinger,aaronson2021new,coladangelo2021hidden,chiesa2022post}), we are not using alternating measurements to estimate success probability, but rather turning it directly into a security game. As far as we know, this direction has never been investigated before. 

In the $g$-alternating measurement games \Cref{fig:alternate-measure}, a challenger first prepares a uniform superposition over all challenges  $\ket{{\mathbbm 1}}_{\bf X} = \frac{1}{\sqrt{N}} \sum_x \ket x$; it then measures ${\bf X}$ together with an adversary's register ${\bf A}$; for each round in $1, 2, \cdots, g$: 
\begin{itemize}
    \item If the current round is odd, the challenger applies the following projection over ${\bf XA}$: 
    \begin{align*}
        \CP^H_0 = \sum_{x \in [N]} \ket x \bra x \otimes (U^H_{H(x)})^\dagger V^H_x U^H_{H(x)} \quad\text{ and }\quad \CP^H_1 = {\bf I}_{\bf XA} - \CP^H_0. 
    \end{align*}
    In other words, $\CP^H_0$ is a controlled projection. If the control is $x$, it will run $\As$ on input $H(x)$ (corresponding to $U^H_{H(x)}$), project into $\As$'s winning (corresponding to $V^H_x$) and undo the computation (which is $(U^H_{H(x)})^\dagger$).
    
    \item If the current round is even, the challenger applies $\IsUniform$ over ${\bf X}$:
    \begin{align*}
        \IsUniform_0 = \ket{{\mathbbm 1}} \bra{{\mathbbm 1}}_{\bf X} \otimes \mathbb{I}_{\bf A} \quad\text{ and }\quad   \IsUniform_1 = {\bf I} - \IsUniform_0. 
    \end{align*}
\end{itemize}
Finally, the winning condition is met when all the measurement outcomes are $0$s. 

\begin{figure}[!hbt]
    \centering
    \begin{gamespec}
    \begin{itemize}
        \item The challenger prepares an equal superposition $\ket{{\mathbbm 1}}_{\bf X} = \frac{1}{\sqrt{N}} \sum_x \ket x$. The whole quantum system over ${\bf XA}$ at the start of the game is $\ket{{\mathbbm 1}}_{\bf X} \ket{\sigma_H, 0^L}_{\bf A}$.
        \item For each round $i \in [g]$, the challenger applies binary measurements over ${\bf XA}$. Let the result be $b_i$. 
        \begin{enumerate}
            \item If it is odd round, it applies $(\CP^H_0, \CP^H_1)$. 
            \item If it is even round, it applies $(\IsUniform_0, \IsUniform_1)$. 
        \end{enumerate}
        \item The adversary wins if and only if $b_1 = b_2 = \cdots = b_g = 0$. 
    \end{itemize}
    \end{gamespec}
    \caption{Alternating Measurement Games for OWFs.}
    \label{fig:alternate-measure}
\end{figure}

The evolution of quantum states in the alternating measurement has a nice form. 
Marriot and Watrous showed that for a eigenvector $\ket{\phi_p}_{\bf A}$, when starting with $\ket{{\mathbbm 1}}_{\bf X}\ket{\phi_p}_{\bf A}$, the state evolves to either $p^{g/2} \ket{\psi_p}_{\bf XA}$ in odd rounds or $p^{g/2} \ket{{\mathbbm 1}}_{\bf X}\ket{\phi_p}_{\bf A}$ in even rounds. Here $\ket{\psi_p}$ is some quantum we do not need to write down explicitly, but importantly it is the same for all even rounds.  For a starting state $\ket{\sigma_H, 0^L} = \sum_{p} \alpha_p \ket{\phi_p}$, the state evolves as in \Cref{fig:evloves_in_MW}. 
\ifllncs 
\begin{figure}[!hbt]
\centering
    \resizebox{\textwidth}{!}{%
			\begin{tikzpicture}[->,shorten >=2pt]
			    \tikzstyle{every node}=[font=\Large]
			    \node (first) [rectangle, draw=none] {$\sum_{p} \alpha_p \ket{{\mathbbm 1}}_{\bf X} \ket{\phi_p}_{\bf A}$};
			    \node (second) [rectangle, below right=of first] {$\sum_{p} p^{1/2} \alpha_p \ket{\psi_p}_{\bf XA}$};
			    \node (third) [rectangle, draw=none, above right=of second] {$\sum_{p} p \alpha_p \ket{{\mathbbm 1}}_{\bf X} \ket{\phi_p}_{\bf A}$};
			    \node (fourth) [rectangle, below right=of third] {$\sum_{p} p^{3/2} \alpha_p \ket{\psi_p}_{\bf XA}$};
			    \node (fifth) [rectangle, above right=of fourth] {$\cdots$};
			    \draw[->] (first) -- (second) node [midway, above right] {$\CP^H_0$};
			    \draw[->] (second) -- (third) node [midway, below right] {$\IsUniform_0$};
			    \draw[->] (third) -- (fourth) node [midway, above right] {$\CP^H_0$};
			    \draw[->] (fourth) -- (fifth) node [midway, below right] {$\IsUniform_0$};
			\end{tikzpicture}
        }
        \caption{Evolution in Alternating Measurement Games.}
        \label{fig:evloves_in_MW}
\end{figure}
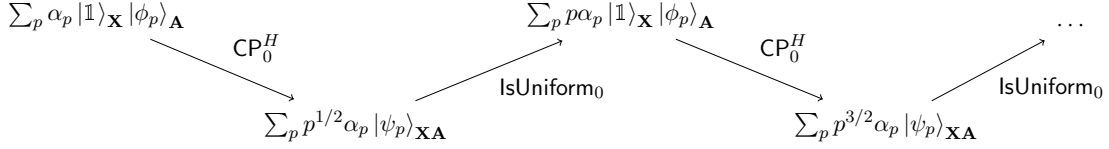
\else 
\begin{figure}[!hbt]
\centering
    \resizebox{0.9\textwidth}{!}{%
			\begin{tikzpicture}[scale=0.8,->,shorten >=2pt]
			    \node (first) [rectangle, draw=none] {$\sum_{p} \alpha_p \ket{{\mathbbm 1}}_{\bf X} \ket{\phi_p}_{\bf A}$};
			    \node (second) [rectangle, below right=of first] {$\sum_{p} p^{1/2} \alpha_p \ket{\psi_p}_{\bf XA}$};
			    \node (third) [rectangle, draw=none, above right=of second] {$\sum_{p} p \alpha_p \ket{{\mathbbm 1}}_{\bf X} \ket{\phi_p}_{\bf A}$};
			    \node (fourth) [rectangle, below right=of third] {$\sum_{p} p^{3/2} \alpha_p \ket{\psi_p}_{\bf XA}$};
			    \node (fifth) [rectangle, above right=of fourth] {$\cdots$};
			    \draw[->] (first) -- (second) node [midway, above right] {$\CP^H_0$};
			    \draw[->] (second) -- (third) node [midway, below right] {$\IsUniform_0$};
			    \draw[->] (third) -- (fourth) node [midway, above right] {$\CP^H_0$};
			    \draw[->] (fourth) -- (fifth) node [midway, below right] {$\IsUniform_0$};
			\end{tikzpicture}
        }
        \caption{Evolution in Alternating Measurement Games.}
        \label{fig:evloves_in_MW}
\end{figure}
\fi 

In the above figure, we do not normalize the quantum states. Their $\ell$-2 norms are then the probability of having all $0$ outcomes until the current round. 
Thereby, $\As$'s winning probability in the $g$-alternating measurement games is exactly $\mathbb{E}_{H}[\p^g] = \mathbb{E}_H[\sum_p |\alpha_p|^2 p^g]$. The remaining step is to bound security in the alternating measurement game.

To conclude, in this paragraph, we show that if a non-uniform algorithm has advantage $\mathbb{E}[\p]$ for the OWF game, it has advantage $\mathbb{E}[\p^g]$ for the $g$-alternating measurement game. It is an analog of the reduction with classical advice: non-uniform security to non-uniform multi-instance security shown in  \Cref{fig:reduction_OWF} (step 1). The most notable benefit is that the new reduction does not need to do rewindings\footnote{One may argue that the game itself does the rewinding, as alternating measurements can be viewed as a way of repairing quantum programs \cite{chiesa2022post}}; as a consequence, neither majority vote nor tomography is required. This is the main reason we obtain tight bounds in this work. 

\paragraph{Step 2: Removing Advice in the Alternating Measurement Games.}
This part is similar to step 2 in \Cref{fig:reduction_OWF}. 
In the previous section, we show a reduction from non-uniform advantage for the OWF game to non-uniform advantage for the $g$-alternating measurement game. However, we did not remove non-uniformity, which is the most troublesome part. We can set $g = S$ and pay a loss of $2^{-S}$ (which comes from a random guess of quantum advice). Thus, if a $S$-qubit non-uniform $\As$ has advantage $\mathbb{E}[\p]$, there must exist a {\bf uniform} quantum algorithm with advantage $2^{-S} \mathbb{E}[\p^S]$ in the $S$-alternating measurement game. The loss $2^{-S}$ will diminish as  $\mathbb{E}[\p^S]$ is also exponential in $S$ in most cases.
In the following paragraph, we only need to consider security of uniform algorithms.

\paragraph{Step 3: Security in the Alternating Measurement Games.}
Recall $b_1, b_2, \cdots, b_g$ are the binary outcomes in the alternating measurement game \Cref{fig:alternate-measure}. 
 Since $\mathbb{E}_H[\p^g] = \Pr[b_1 = \cdots = b_g = 0]$, we have: 
\begin{align*}
    \mathbb{E}_H[\p^g] = \prod_{i=1}^g \Pr[b_i = 0\,|\, b_{<i} = 0]. 
\end{align*}
We bound the conditional probability $\Pr[b_i = 0\,|\, b_{<i} = 0]$ for each individual $1 \leq i \leq g$; i.e., an adversary wins the $i$-th round, conditioned on its winning all the previous rounds. 

We observe that the conditional probability is monotonically non-decreasing. This is due to $\Pr[b_i = 0\,|\, b_{<i} = 0] = {\Pr[b_{< i + 1} = 0]}/{\Pr[b_{<i} = 0]} =  {\mathbb{E}_H[\p^i]}/{\mathbb{E}_H[\p^{i-1}]}$. The monotonicity of the conditional probabilities follows by Jensen's inequality. 
Therefore, we only need to bound the last term, $\varepsilon_g := \Pr[b_g = 0\,|\, b_{<g} = 0]$; and $\mathbb{E}_H[\p^g] \leq \varepsilon_g^g$.
Finally, we show $\varepsilon_g$ can be bounded using the existing theorem in \cite{chung2020tight}. 

We briefly recap the idea in \cite{chung2020tight}. 
To prove security in the multi-instance setting, 
\cite{chung2020tight} indeed prove a stronger statement. They show that for any $P$-quantum-query uniform algorithm $f$ and $T$-query uniform $\As$, the probability that $\As$ wins the OWF game conditioned on $f^H = 0$ is still bounded by $O(P + T^2)/N$. This model is later named as $P$-Bit-Fixing-QROM (or $P$-BF-QROM) by \cite{guo2021unifying}. One can view the algorithm $f$ as quantumly fixing $P$ coordinates of a random oracle, and the security says regardless of $f$'s behavior, as long as it is query bounded, the online algorithm only has limited advantages of inverting a uniformly random image.

When $g$ is odd, the measurement on the $g$-th round is $\CP^H_0, \CP^H_1$, in which an outcome $0$ corresponds to winning the OWF game. 
$\varepsilon_g$ is bounded by the advantage in the $P$-BF-QROM when $P \approx g T$: since we can set $f$ as a quantum algorithm that does alternating measurements for the first $g-1$ rounds, and $\As$ is the algorithm that plays the OWF game. Thus, $\varepsilon_g = O((gT + T^2)/N$ for odd $g$.

When $g$ is even,  the measurement on the $g$-th round is $(\IsUniform_0, \IsUniform_1)$. This step, unlike the case for odd $g$, has less physical meaning and we do not know how to bound it directly. But fortunately, we have $\epsilon_g \leq \epsilon_{g+1} = O((gT + T^2)/N)$ since we prove the conditional probability is non-decreasing and $\varepsilon_{g+1}$ is easy to bound for $g+1$ being odd. 

Therefore, we can bound $\varepsilon_g$ as well as $\mathbb{E}[\p^g]$ for all positive $g$ .

\paragraph{Achieving Non-Uniform Security.}
We combine all the steps above and achieve non-uniform security with quantum advice for OWG games (\Cref{fig:recap}). 
\begin{itemize}
    \item For any non-uniform quantum algorithm with $S$-qubits of advice and $T$ queries, let its advantage be $\delta = \mathbb{E}[\p]$. 
    \item(Step 1.) Its advantage in the $g$-alternating measurement game is $\mathbb{E}[\p^g]$. 
    \item(Step 2.) By guessing the quantum advice, we can further remove the non-uniformity. There exists a uniform algorithm in the $g$-alternating measurement game with advantage $2^{-S} \mathbb{E}[\p^g]$. 
    \item(Step 3.) By setting $g = S$ and $P = ST$, we can bound $2^{-S} \mathbb{E}[\p^g]$ by $O((ST + T^2)/N)^S$. 
\end{itemize}
Combining all the steps above, we have:
\begin{align*}
    \delta = \mathbb{E}[\p] \leq (\mathbb{E}[\p^S])^{1/S} \leq \left(2^S \cdot O((ST+T^2)/N)^S\right)^{1/S} = O((ST+T^2)/N),
\end{align*}
finishing the proof for \Cref{thm:simple_OWF}. 

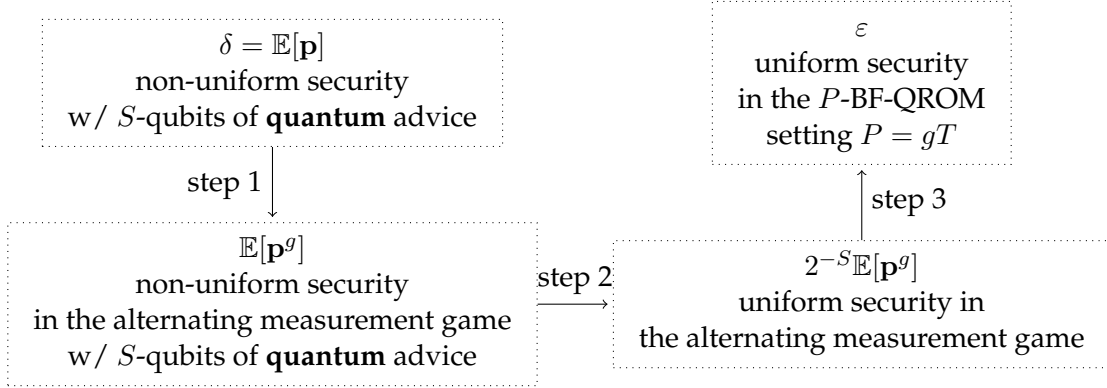
\begin{figure}[!hbt]
\centering
			\begin{tikzpicture}[scale=0.9,->,shorten >=2pt]
			    \node (first) [rectangle, dotted, draw=black]  {\begin{tabular}{c} $\delta = \mathbb{E}[\p]$ \\  non-uniform security \\ w/ $S$-qubits of {\bf quantum} advice \end{tabular}};
			    \node (second) [rectangle, dotted, draw=black, below=of first] {\begin{tabular}{c} $\mathbb{E}[\p^g]$ \\  non-uniform security \\ in the alternating measurement game \\ w/ $S$-qubits of {\bf quantum} advice \end{tabular}};
			    \node (third) [rectangle, dotted, draw=black, right=of second]{\begin{tabular}{c} $2^{-S} \mathbb{E}[\p^g]$ \\ uniform security in \\  the alternating measurement game \end{tabular}};
			    \draw[->] (first) -- (second) node [midway, left] {step 1};
			    \draw[->] (second) -- (third) node [midway, above] {step 2};
			    \node (fourth) [rectangle, dotted, draw=black, above=of third] {\begin{tabular}{c} $\varepsilon$ \\ uniform security \\ in the $P$-BF-QROM \\ setting $P = gT$ \end{tabular}};
			    \draw[->] (third) -- (fourth) node [midway, right] {step 3};
			\end{tikzpicture}
        \caption{Our reduction (comparing with the reduction in \cite{chung2020tight} \Cref{fig:reduction_OWF}).}
        \label{fig:recap}
\end{figure}

We can extend the above reduction to a general framework for non-uniform security.
As long as the advantage of a game in the $P$-BF-QROM can be bounded by $\varepsilon$, we can establish its non-uniform security by $O(\varepsilon)$. For decision games, instead of setting $g = S$, we choose a different $g$, but other ideas roughly follow. Please refer to \Cref{sec:main_theorem} for more details. 

\paragraph{Separation of Classical and Quantum Advice.}

Our separation is based on the recent work by Yamakawa and Zhandry \cite{yamakawa2022verifiable}. They show, relative to a random oracle, there exists a one-way function such that: (1) it is hard for any polynomial-query (or even subexponential-query) classical algorithm to invert a challenge image $y$; (2) a quantum algorithm (Yamakawa-Zhandry algorithm) can invert any $y$ with certainty. We observe that the quantum algorithm has the additional fascinating property: the algorithm makes non-adaptive queries, and the queries are even independent of $y$; then, post-processing depending on $y$ reveals a pre-image.

The OWF in \cite{yamakawa2022verifiable} is an example of separating classical and quantum advice for $S$ being some subexponential function and $T = 0$. 
When advice is quantum, the advice can be the queries made by the Yamakawa-Zhandry algorithm because these queries are independent of $y$. The winning probability with quantum advice is equal to that of the Yamakawa-Zhandry algorithm, which is arbitrarily close to $1$. When only classical advice is given, we show that any algorithm only knows at most $\tilde{O}(S)$ positions of the random oracle, based on a theorem from \cite{coretti2018random}; as no more queries are allowed, most images $y$ are not queried and thus can never be inverted.

\section{Preliminaries}
\label{sec:prelim}

We assume readers are familiar with the basics of quantum information and computation. All backgrounds on quantum information can be found in \cite{nielsen2002quantum}. 

\subsection{Quantum Random Oracle Model}
In the quantum random oracle model, a hash function is modeled as a random classical function $H$. The function $H$ is sampled at the beginning of any security game and then gets fixed. Oracle access to $H$ is defined by a unitary $U_H: \ket{x, y} \to \ket{x, y + H(x)}$. 
A quantum oracle algorithm with oracle access to $H$ is then denoted by a sequence of unitary $U_1$, $U_H$, $U_2$, $U_H$, $\cdots$, $U_T$, $U_H$, $U_{T+1}$ followed by a computational basis measurement, where $U_i$ is a local unitary operating on the algorithm's internal register. The number of queries, in this case, is $T$ --- the number of $U_H$ calls.

\ifllncs 
\input{helpfullemma_llncs}
\else 
\subsection{Other Useful Lemmas}

We use the lemmas in this section to prove bounds in the alternating measurement games (\Cref{sec:main_theorem}). Readers can safely skip and return to this section for understanding proofs in \Cref{sec:main_theorem}. 

\begin{lemma}\label{lem:prob_reweight}
Let $N$ be a positive integer and $p_1, \cdots, p_N \in \mathbb{R}^{\geq 0}$. 
Let $\alpha_1, \cdots, \alpha_N$ be a distribution over $[N]$: i.e., $\alpha_i \in [0, 1]$ and $\sum_{i \in [N]} \alpha_i = 1$. 

Assume $\mu := \sum_{i \in [N]} \alpha_i p_i > 0$. 
Let $\beta_1, \cdots, \beta_N$ be another distribution over $[N]$: $\beta_i := \alpha_i p_i / \mu$. The following holds: 
\begin{align*}
    \sum_{i \in [N]} \beta_i p_i \geq \sum_{i \in [N]} \alpha_i p_i. 
\end{align*}
\end{lemma}

\begin{proof}
Let $\mb{X}$ be a random variable that takes value $p_i$ w.p. $\alpha_i$. It is easy to see that $\mbb{E}[\mb{X}] = \sum_{i} \alpha_i p_i$ and 
$\mbb{E}[\mb{X}^2] = \sum_{i} \alpha_i p_i^2$.

Since we assume $\mu = \mbb{E}[\mb{X}] > 0$, we rewrite the inequality as follows: 
\begin{align*}
    \sum_{i} \alpha_i p_i^2 \geq \left(\sum_{i} \alpha_i p_i\right)^2.
\end{align*}
The lemma holds by observing that L.H.S. is $\mbb{E}[\mb{X}^2]$, R.H.S. is $\mbb{E}[\mb{X}]^2$ and the fact that $\mb{Var}[\mb{X}] := \mbb{E}[\mb{X}^2] - \mbb{E}[\mb{X}]^2 \geq 0$. 
\end{proof}

\begin{lemma}\label{lem:prob_monotone}
Let $N$ be a positive integer and $p_1, \cdots, p_N \in \mathbb{R}^{\geq 0}$. 
Let $c_1, \cdots, c_N$ be a distribution over $[N]$. Assume $\sum_{i \in [N]} c_i p_i > 0$. 
Define $S_k$ for every integer $k \geq 1$: 
\begin{align*}
    S_k = \frac{\sum_{i \in [N]} c_i p_i^{k}}{\sum_{i \in [N]} c_i p_i^{k-1}}. 
\end{align*}
Then $\{S_k\}_{k \geq 1}$ is monotonically non-decreasing. %
\end{lemma}
\begin{proof}
We fix any integer $k \geq 1$. 
Let $\alpha_i = c_i p_i^{k-1} / (\sum_i c_i p_i^{k-1})$. It it easy to see that $S_k = \sum_{i} \alpha_i p_i$. 

Let $\beta_i = \alpha_i p_i / \mu$ where $\mu = \sum_i \alpha_i p_i$. We have 
\begin{align*}
    \beta_i &=  \alpha_i p_i / \mu \\
            &= \frac{c_i p_i^{k}}{\sum_i c_i p_i^{k-1}  \cdot \mu }\\
            &= \frac{c_i p_i^{k}}{\sum_i c_i p_i^{k-1} \cdot \left( \sum_i c_i p_i^{k} / (\sum_i c_i p_i^{k-1})  \right) } \\
            &= \frac{c_i p_i^{k}}{\sum_i c_i p_i^{k}}. 
\end{align*}
Therefore, $S_{k+1} = \sum_i \beta_i p_i$. 
By \Cref{lem:prob_reweight}, $S_{k+1} =  \sum_i \beta_i p_i \geq  \sum_i \alpha_i p_i = S_{k}$. 
\end{proof}

\begin{lemma}[Jensen's inequality]\label{lem:jesen}
Let $N, g$ be two positive integers and $p_1, \cdots, p_N \in \mathbb{R}^{\geq 0}$. 
Let $c_1, \cdots, c_N$ be a distribution over $[N]$. Assume $\sum_{i \in [N]} c_i p_i > 0$. 
If the following holds
\begin{align*}
    {\sum_{i \in [N]} c_i p_i^{g}} \leq \delta^g, 
\end{align*}
then $\sum_{i \in [N]} c_i p_i \leq \delta$. 
\end{lemma}

\fi

\section{\texorpdfstring{$(S, T)$ Quantum Algorithms and Games in the QROM}{(S, T) Quantum Algorithms and Games in the QROM}}
\label{sec:gamesandalgorithms}

In this work, we consider non-uniform algorithms against games in the QROM. We start by defining $(S, T)$ non-uniform quantum algorithms with either $S$ classical bits of advice or $S$ qubits of advice. The definitions below more or less follow definitions in \cite{chung2020tight} but are adapted for our setting.

\begin{definition}[$(S, T)$ Non-Uniform Quantum Algorithms in the QROM]\label{def:STalgo}
A $(S, T)$ non-uniform quantum algorithm with {\bf classical} advice in the QROM is modeled by a collection $\{s_H\}_{H:[N] \to [M]}$ and $\{U_\inp\}_{\inp}$: for every function $H$, $s_H$ is a piece of $S$-bit advice and $U^H_{\inp}$ is a unitary that calls the oracle $H$ at most $T$ times. 

A $(S, T)$ non-uniform quantum algorithm with {\bf quantum} advice in the QROM is modeled by a collection $\{\ket{\sigma_H}\}_{H}$ and $\{U_\inp\}_{\inp}$: for every function $H$, $\ket{\sigma_H}$ is a piece of $S$-qubit advice and $U^H_{\inp}$ is a unitary that calls the oracle $H$ at most $T$ times. 

Similarly, we denote a {\bf uniform} quantum algorithm by a collection of unitaries $\{U_{\inp}\}_{\inp}$: it is a non-uniform quantum algorithm satisfying $\ket{\sigma_H} = \ket{0^S}$ for all $H$.

\medskip

When the algorithm is working with oracle access to $H$, its initial state is $\ket{s_H} \ket{0^L}$ or $\ket{\sigma_H} \ket{0^L}$, respectively. On input $\inp$, it applies $U^H_{\inp}$ on the initial state and measures its internal register in the computational basis. 
\end{definition}

Since we are working in the idealized model, we require neither $L$ nor the size of the unitary $U_{\inp}$ to be polynomially bounded. In the rest of the work, we will focus on non-uniform algorithms with quantum advice as our new reduction works for both cases. Therefore, `non-uniform algorithms' denotes `non-uniform algorithms with quantum advice'. 

\begin{remark}
We can assume quantum advice is a {\bf pure} state. Due to convexity, the optimal non-uniform algorithm can always have advice as a pure state. If the advice is a mixed state and achieves a winning probability $p$, there always exists a pure state that achieves a winning probability at least $p$.
\end{remark}

\medskip

Next, we define games in the QROM. 
\begin{definition}[Games in the QROM]\label{def:game}
A game $G$ in the QROM is specified by two classical algorithms $\samp^H$ and $\ver^H$: 
\begin{itemize}
    \item $\samp^H(r)$: it is a deterministic algorithm that takes uniformly random coins $r \in \mc{R}$ as input, and outputs a challenge $\ch$.
    \item $\ver^H(r, \ans)$: it is a deterministic algorithm that takes the same random coins for generating a challenge and an alleged answer $\ans$, and outputs $b$ indicating whether the game is won ($b = 0$ for winning). 
\end{itemize}
Let $T_\samp$ be the number of queries made by $\samp$ and $T_\ver$ be the number of queries made by $\ver$. 

\medskip

For a fixed $H$ and a quantum algorithm $\As$, the game $G^H_{\As}$ %
is executed as follows:
\begin{itemize}
    \item A challenger $\Cs$ samples $\ch \gets \samp^H(r)$ using uniformly random coins $r$.
    \item A (uniform or non-uniform) quantum algorithm $\As$ has oracle access to $H$, takes $\ch$ as input and outputs $\ans$. We call $\As$ an online adversary/algorithm. 
    \item $b \gets \ver^H(r, \ans)$ is the game's outcome. 
\end{itemize}
\end{definition}

\begin{remark}
In the above definition, a quantum algorithm makes at most $T$ oracle queries to $H$. However, in some particular games, the algorithm can not get access to $H$. One famous example is Yao's box, in which an adversary is given a challenge input $x$ and the goal is to output $H(x)$. The adversary can query $H$ on any input except $x$ (otherwise, the game is trivial). The definition \Cref{def:game} does not capture this case. Nonetheless, we will stick with the current definition. For the special case when an algorithm has access to a different oracle $H'$, the technique in this work extends as well. This extension requires a similar definition of games (Definition 3.3) in \cite{chung2020tight}. 
\end{remark}

Let us warm up by having a close look at the following examples. 
\begin{example}\label{ex:OWFandPRG}
The first example is function inversion (or OWFs) $G_{\sf OWF}$. $r = x \in [N]$ is a uniformly random pre-image and $\ch := H(x)$. The goal is to find a pre-image of $\ch$. The verification procedure takes $r = x$ and $\ans = x'$, it outputs $0$ (winning) if and only if $x'$ is a pre-image of $H(x)$. 

The other example $G_{\sf PRG}$ is to distinguish images of PRG from a uniformly random element. In this example, $r$ consists of $(b, x, y)$ where $b$ is a single bit, $x$ is a uniformly random pre-image in $[N]$ and $y$ is a uniformly random element in $[M]$. The challenge $\ch$ is $H(x)$ if $b = 0$, otherwise $\ch = y$.  The goal is to distinguish whether an image of a random input or a random element in the range is given. The verification procedure takes $r = (b, x, y)$ and $\ans = b'$, it outputs $0$ if and only if $b = b'$. 
\end{example}

\begin{definition}
We say a game $G$ has $\delta(S, T) := \delta$ maximum winning probability (or has security $\delta$, for cryptographic games) against all $(S, T)$ non-uniform quantum adversaries with classical or quantum advice if 
\begin{align*}
    \max_{\As}\Pr_{H}\left[ G^H_\As = 1 \right] \leq \delta,
\end{align*}
where $\max$ is taken over all $(S, T)$ non-uniform quantum adversaries $\As$ with classical or quantum advice, respectively.
\end{definition}

\subsection{Quantum Bit-Fixing Model}
Here we recall a different model called the quantum bit-fixing model. In the following sections, we will relate winning probability of a game $G$ against $(S, T)$ non-uniform quantum algorithms with that in the quantum bit-fixing model (BF-QROM). Since the previous quantum non-uniform bounds require analyzing the quantum bit-fixing model, winning probabilities in the bit-fixing model are already known for many games, and our improved bounds only need a new reduction. The following definitions are adapted from \cite{guo2021unifying}. 

\begin{definition}[Games in the $P$-BF-QROM]\label{def:pbfqrom}
It is similar to games in the standard QROM, except now $H$ has a different distribution. 
    \begin{itemize}
        \item Before a game starts, a quantum algorithm $f$ (having no input) with at most $P$ queries to an oracle is picked and fixed by an adversary. 
        \item {\bf Rejection Sampling Stage:} A random oracle $H$ is picked uniformly at random, then conditioned on $f^H$ outputs $0$. In other words, the distribution of $H$ is defined by a rejection sampling:
        \begin{enumerate}
            \item $H \gets \{f:[N] \to [M]\}$. 
            \item Run $f^H$ and obtain a binary outcome $b$ together with a quantum state $\tau$\footnote{In \cite{guo2021unifying}, they do not need quantum or classical memory $\tau$ shared between $f$ and $\As$. However, this is essential in our proof. Nonetheless, all security proofs in the $P$-BR-QROM work in the stronger setting (with $\tau$ shared between stages).}. \item Restart from step 1 if $b \ne 0$.
        \end{enumerate}
        \item {\bf Online Stage:} The game is then executed with oracle access to $H$, and an algorithm $\Bs$ gets $\tau$.
    \end{itemize}
\end{definition}

A $(P,T)$ algorithm in the $P$-BF-QROM consists of $f$ for sampling the distribution and $\Bs$ for playing the game, with $f$ making at most $P$ queries and $\Bs$ making at most $T$ queries. We also call $\Bs$ an {\bf online} algorithm/adversary. 

We will also consider the following classical analog $P$-BF-ROM only when showing a separation between classical and quantum advice in \Cref{sec:separation}. 
\begin{definition}[Games in the $P$-BF-ROM]\label{def:pbfrom}
It is similar to the above \Cref{def:pbfqrom}, except both $f$ and $\Bs$ can only make classical queries. 
\end{definition}

\begin{definition}
We say a game $G$ has $\nu(P, T) := \nu$ maximum winning probability (or is $\nu$-secure, for cryptographic games) in the $P$-BF-QROM if 
\begin{align*}
    \max_{f, \Bs}\Pr_{{H}}\left[ f^H = 0 \,\wedge\, G^H_\Bs = 1 \right] \leq \nu,
\end{align*}
where $\max$ is taken over all $(P, T)$ quantum adversaries $(f, \Bs)$ with $f$ making at most $P$ queries and $\Bs$ making at most $T$ queries.    
\end{definition}

We know the following two lemmas from \cite{chung2020tight,guo2021unifying}. 
\begin{lemma}[Function Inversion in the $P$-BF-QROM]
The OWF game has $\nu(P, T) = (P + T^2)/\min\{N,M\}$ in the $P$-BF-QROM.
\end{lemma}
See the proof for Lemma 5.2 in \cite{chung2020tight} and Lemma 10 in \cite{guo2021unifying}.  

\begin{lemma}[PRGs in the $P$-BF-QROM]
The game PRG has $\nu(P, T) = 1/2 + \sqrt{(P + T^2)/N}$ in the $P$-BF-QROM.
\end{lemma}
See the proof for Lemma 5.13 in \cite{chung2020tight}.

\section{Games, POVMs and Decomposition of Advice}
\label{sec:notations}

In this section, we will formalize an quantum algorithm's winning probability against a game in terms of POVMs and its corresponding eigenvectors. 

For any game $G$ and algorithm $\As$, let $V^H_r$ be a projection that operates on the register of $\As$. $V^H_r$ project a quantum state into a subspace spanned by basis states $\ket{\ans}\ket{z}$ where $\ver^H(r, \ans) = 1$ and $z$ be any aux input (depending on the size of $\As$'s working register). As an example, for function inversion problem and $r = x$, $V^H_r$ is defined as $\sum_{x': H(x') = H(x), z} \ket{x', z} \bra{x', z}$. 

Then for any non-uniform quantum algorithm $\As = (\{\ket{\sigma_H}\}_H, \{U_\inp\}_{\inp})$, by definition, its probability $\epsilon_\As$ for winning the game $G$ with oracle access to $H$ can be then written as: 
\begin{align*}
    \epsilon_{\As, H} = \frac{1}{|\mc{R}|} \sum_{r \in \mc{R}} \left\Vert V^H_r U^H_{\samp^H(r)} \ket{\sigma_H} \ket{0^L} \right\Vert^2. 
\end{align*}

We define the following projections $P_r^H := \left(U^H_{\samp^H(r)}\right)^\dagger V^H_r U^H_{\samp^H(r)}$. Let $P_H$ be a POVM: 
\ifllncs 
$P_H := \frac{1}{|\mc{R}|} \sum_{r \in \mc{R}} P_r^H$.
\else 
\begin{align*}
    P_H := \frac{1}{|\mc{R}|} \sum_{r \in \mc{R}} P_r^H. 
\end{align*}
\fi 
We can equivalently write $\epsilon_{\As, H}$ in terms of this POVM: $\epsilon_{\As, H} = \langle \sigma_H, 0^L | P^H | \sigma_H, 0^L\rangle$. This is due to:
\begin{align*}
    \epsilon_{\As, H} =&  \frac{1}{|\mc{R}|} \sum_{r \in \mc{R}} \left\Vert V^H_r U^H_{\samp^H(r)} \ket{\sigma_H} \ket{0^L} \right\Vert^2 \\
    =&  \frac{1}{|\mc{R}|} \sum_{r \in \mc{R}} \bra{\sigma_H} \bra{0^L}  P^H_r \ket{\sigma_H} \ket{0^L} \\
    =& \langle \sigma_H, 0^L | P^H | \sigma_H, 0^L\rangle.
\end{align*}

\medskip

Since $P_H$ is a Hermitian matrix and $0 \preceq P_H \preceq \bf{I}$, let $\{\ket{\phi_{H, j}}\}_j$ be the set of eigenbasis for $P_H$ with eigenvalues $\{p_{H, j}\}_j$ between $0$ and $1$. We can decompose $\ket{\sigma_H}\ket{0^L}$ under the eigenbasis: 
\begin{align*}
    \ket{\sigma_H}\ket{0^L} = \sum_{i} \alpha_{H, i} \ket{\phi_{H, i}}. 
\end{align*}
Therefore, $\epsilon_{\As,H}$ can be written in terms of $\alpha_{H,i}$ and $p_{H, i}$: $\epsilon_{\As, H} = \sum_{i} |\alpha_{H, i}|^2 \cdot p_{H, i}$. This is because:  
\begin{align*}
    \epsilon_{\As,H} = \langle \sigma_H, 0^L | P^H | \sigma_H, 0^L\rangle = \sum_{i} |\alpha_{H, i}|^2 \cdot p_{H, i}. 
\end{align*}

With all the above discussions, we conclude our lemma below.
\begin{lemma}
Let $G$ be a game and $\As = (\{\ket{\sigma_H}\}_H, \{U_\inp\}_{\inp})$ be any non-uniform quantum algorithm. Let $P_H$ be the corresponding POVMs for function $H$. Let $\{\ket{\phi_{H, j}}\}_j$ be the set of eigenbasis for $P_H$ with eigenvalues $\{p_{H, j}\}_j$. 

For each $H$, write $\ket{\sigma_H}\ket{0^L}$ as $\sum_{i} \alpha_{H, i} \ket{\phi_{H, i}}$. Let $\epsilon_\As$ be the winning probability of $\As$, when $H$ is drawn uniformly at random. Then
\begin{align*}
    \epsilon_\As = \mathbb{E}_{H}\left[\sum_{i} |\alpha_{H,i}|^2 \cdot p_{H, i}\right] = \frac{1}{N^M} \sum_{H} \sum_i |\alpha_{H, i}|^2 \cdot p_{H,i}. 
\end{align*}
\end{lemma}

\section{Non-Uniform Lower Bounds via Alternating Measurements}
\label{sec:main_theorem}

In this section, we prove the following theorem:
\begin{theorem}\label{thm:bftononuniform}
Let $G$ be any game with $T_\samp, T_\ver$ being the number of queries made by $\samp$ and $\ver$. 
For any $S, T$, let $P = S (T + T_\ver + T_\samp)$. 

If $G$ has security $\nu(P, T)$ in the $P$-BF-QROM, then it has security (maximum winning probability) $\delta(S, T) \leq 2 \cdot \nu(P, T)$ against $(S, T)$ non-uniform quantum algorithms with quantum advice. 

It also has security
\begin{align*}
    \delta(S, T) \leq \min_{\gamma > 0} \left\{\nu(P/\gamma, T) + \gamma \right\}
\end{align*}
against $(S, T)$ non-uniform quantum algorithms with quantum advice. 
\end{theorem}
As a special case of the second result, when $G$ is a decision game and is $\nu(P, T) = \frac{1}{2} + \nu'(P, T)$ secure in the $P$-BF-QROM, then it has security
\begin{align*}
    1/2 + \min_{\gamma > 0} \left\{\nu'(P/\gamma, T) + \gamma \right\}
\end{align*}
against $(S, T)$ non-uniform quantum algorithms with quantum advice.

\medskip

The section is organized as follows: in the first subsection, we introduce a new multi-instance game, via the so-called alternating measurement games, the idea of alternating measurement was used in witness preserving amplification of QMA (\cite{marriott2005quantum}); in the next subsection, we elaborate on behaviors of any non-uniform quantum algorithm in the alternating measurement game; then we show that upper bounds (of success probabilities) in the bit-fixing model give rise to the probability of {\bf uniform} quantum algorithms in the alternating measurement game; finally in the last subsection, we give the proof for our main theorem.

\subsection{Multi-Instance via Alternating Measurements}

For a game $G$ and a quantum non-uniform algorithm $\As = (\{\ket{\sigma_H}\}_H, \{U_\inp\}_{\inp})$, we start by recalling the following notations as in \Cref{sec:notations}: $P^H_r, P_H, \{\ket{\phi_{H, j}}\}_j$ and $\{p_{H, i}\}_{j}$. 
Let $\rA$ be the register that $\As$ operates on. 
The following controlled projection (as defined in \cite{zhandry2020schrodinger}) will be used heavily in this section. 

\begin{definition}[Controlled Projection]
The controlled projection for a game $G$ and a quantum algorithm $\As$ is the following: for every $H$, the controlled projection is the measurement $\CP^H = (\CP_0^H, \CP_1^H)$:
\begin{align*}
    \CP_0^H = \sum_{r \in \mc{R}} {\ket r} {\bra r}_\rR \otimes P^H_r \quad\text{ and }\quad \CP_1^H = \sum_{r \in \mc{R}} \ket r {\bra r}_{\rR} \otimes (\identity_\rA - P^H_r).
\end{align*}
\end{definition}
Here $\CP^H$ operates on registers $\mc{RA}$ where $\mc{R}$ are registers storing random coins and $\mc{A}$ are $\As$'s working registers. 

Similarly, we define the following projection $\IsUniform = (\oneproject \otimes \identity_\rA, (\identity_\rR - \oneproject) \otimes \identity_\rA)$ over the same register as $\CP^H$ where $\ket{\mathbbm{1}_\mc{R}}$ is a uniform superposition over $\mc{R}$: i.e., $\ket{\mathbbm{1}_\mc{R}} = \frac{1}{|\mc{R}|} \sum_r \ket{r}$. We denote $\oneproject \otimes \identity_\rA$ by $\IsUniform^0$ and $(\identity-\oneproject \otimes \identity_\rA)$ by $\IsUniform^1$.

\medskip

Now, We are ready to describe the new game via alternating measurements: 
\begin{definition}[Multi-Instances via Alternating Measurments]\label{def:misviaam}
Fix a game $G$ and an integer $k \geq 1$. 
A uniformly random $H$ is sampled at the beginning. For a (potentially non-uniform) quantum algorithm $\As$, the multi-instance game $G^{\otimes k}$ is defined and executed as follows:
\begin{itemize}
    \item A challenger $\Cs$ initializes a new register $\ket{\mathbbm{1}_{\mc{R}}}_\rR$ and controls $\As$'s register $\rA$.  
    \item It repeats the following procedures $k$ times, for $i = 1, \cdots, k$: 
    \begin{itemize}
        \item If the current stage $i$ is odd, $\Cs$ applies $\CP^H$ on $\rR \rA$ and obtains a measurement outcome $b_i$.
        \item If the current stage $i$ is even, $\Cs$ applies $\IsUniform$ on $\rR \rA$ and obtains a measurement outcome $b_i$.
    \end{itemize}
    \item The game is won if and only if $b_1 = b_2 = \cdots = b_{k} = 0$.
\end{itemize}
\end{definition}

With this alternating measurement game, we describe the following theorem that relates the winning probability of a (non-uniform) $\As$ in the game $G$ and that of $\As$ in the corresponding alternating measurement game $G^{\otimes k}$. 
\begin{theorem}\label{thm:misprob}
Let $G$ be a game and $\As = (\{\ket{\sigma_H}\}_H, \{U_\inp\}_{\inp})$ be any non-uniform quantum algorithm for $G$. Let $P_H$ be the corresponding POVMs for function $H$. Let $\{\ket{\phi_{H, j}}\}_j$ be the set of eigenbasis for $P_H$ with eigenvalues $\{p_{H, j}\}_j$. 

For each $H$, write $\ket{\sigma_H}\ket{0^L}$ as $\sum_{i} \alpha_{H, i} \ket{\phi_{H, i}}$. Let $\epsilon^{\otimes k}_\As$ be the winning probability of $\As$ in the alternating measurement game $G^{\otimes k}$, when $H$ is drawn uniformly at random. Then
\begin{align*}
    \epsilon^{\otimes k}_\As = \frac{1}{N^M} \sum_H \sum_{i} |\alpha_{H,i}|^2 \cdot p^k_{H, i}. 
\end{align*}
\end{theorem}
\ifllncs 
We leave the explanation of the theorem to the appendix (the proof of \Cref{lem:amprogress}) since it is similar to the analysis of QMA amplification \cite{marriott2005quantum} and quantum traitor tracing \cite{zhandry2020schrodinger}. We do not considered the proof as our main contribution. Nonetheless, we believe that the proof inspires our analysis for $\epsilon_\As^{\otimes k}$, which together with the new multi-instance reduction is considered the main contribution of this work. 
\else 
We leave the explanation of the theorem to the next section (the proof of \Cref{lem:amprogress}) since it is similar to the analysis of QMA amplification \cite{marriott2005quantum} and quantum traitor tracing \cite{zhandry2020schrodinger}. We do not considered the proof as our main contribution. Nonetheless, we believe that the proof inspires our analysis for $\epsilon_\As^{\otimes k}$, which together with the new multi-instance reduction is considered the main contribution of this work. 
\fi 

By \Cref{lem:jesen}, we can easily conclude that any upper bound on $\As$'s success probability in $G^{\otimes k}$ yields an upper bound on its winning probability in $G$. The proof of the following lemma easily follows from \Cref{lem:jesen}. 
\begin{lemma}\label{lem:non-uniform_to_mis}
Fix a game $G$ and an integer $k \geq 1$. Let $\epsilon_\As$ be the success probability of (uniform or non-uniform) $\As$ in $G$ and $\epsilon^{\otimes k}_\As$ be that of $\As$ in the alternating measurement game $G^{\otimes k}$. Then $\epsilon_\As \leq \left(\epsilon_\As^{\otimes k}\right)^{1/k}$. 
\end{lemma}
Thereby, to bound $\epsilon_\As$, it is enough to bound $\epsilon_{\As}^{\otimes k}$ for some appropriate positive integer $k$.

\ifllncs 
\else 
\subsection{\texorpdfstring{Characterization of Alternating Measurements and Proof of \Cref{thm:misprob}}{Characterization of Alternating Measurements}}

Fixing a function $H$, the intial internal register $\rA$ of $\As$ is $\ket{\sigma_H}\ket{0^L} = \sum_{i} \alpha_{H, i} \ket{\phi_{H, i}}$. Let us define the following states $\ket{v_{H, i}^0}, \ket{v_{H, i}^1}, \ket{w_{H, i}^0}, \ket{w_{H, i}^1}$ (for convenience, we ignore $H$ in the subscripts in the analysis below). We will also ignore $H$ for other notations like $P^H_r, \ket{\phi_{H, i}}, p_{H, i}$ as our analysis does not depend on $H$ and the final conclusion follows by taking expectation over uniformly random functions $H$. Instead, we are using $P_r := P^H_{r}, \ket{\phi_i}:=\ket{\phi_{H, i}}, p_i := p_{H, i}$ in the analysis. 

\begin{enumerate}
    \item $\ket{w_i^0} = \frac{1}{\sqrt{p_{i} |\mc{R}|}} \sum_r \ket{r} P_r \ket{\phi_{i}}$. 
    
    It is easy to verify that it has norm $1$:
    \begin{align*}
        \langle w_i^0 | w_i^0\rangle = \frac{1}{p_i |\mc{R}|} \sum_r \langle \phi_i | P_r | \phi_i \rangle 
         = \frac{1}{p_i |\mc{R}|}   \langle \phi_i | (\sum_r P_r) | \phi_i \rangle 
         =\frac{p_i |\mc{R}| }{p_i |\mc{R}|} \,= 1.
    \end{align*}
    $\CP^H_0 \ket{w_i^0} = \ket{w_i^0}$ and $\CP^H_1 \ket{w_i^0} = 0$. 
    
    After seeing the definition of $\ket{v_i^0}$ and $\ket{v_i^1}$ below, we also observe that $\ket{w_i^0} = \sqrt{p_i} \ket{v_i^0} + \sqrt{1-p_i} \ket{v_i^1}$. 
    
    \item $\ket{w_i^1} = \frac{1}{\sqrt{(1-p_{i}) |\mc{R}|}} \sum_r \ket{r} (\identity_{\rA}-P_r) \ket{\phi_{i}}$. 
    
    Similarly, it has norm $1$, $\CP^H_1 \ket{w_i^1} = \ket{w_i^1}$ and $\CP^H_0 \ket{w_i^1} = 0$.

    \item $\ket{v^0_i} = \ket{\mathbbm{1}}_{\mc{R}} \ket{\phi_{i}} = \sqrt{p_i} \ket{w_i^0} + \sqrt{1 - p_i} \ket{w^1_i}$. 
    
    By the description of the game $G^{\otimes k}$ (\Cref{def:misviaam}), the overall register $\rR \rA$ at the beginning of the game can be written as $\sum_{i} \alpha_{i} \ket{v^0_i}$ (which we will prove below). 
    
    The state has norm $1$, $\IsUniform^0 \ket{v^0_i} = \ket{v^0_i}$ and $\IsUniform^1 \ket{v^0_i} = 0$.
    
    \item $\ket{v_i^1} = \sqrt{1 - p_i} \ket{w_i^0} - \sqrt{p_i} \ket{w^1_i}$. 
    
    We will not use the property of $\ket{v_i^1}$ in the proof and we thus omit all the details here. 
\end{enumerate}

\begin{lemma}\label{lem:amprogress}
For any fixed $H$, for any non-negative integer $k$, the leftover state over $\rR \rA$ conditioned on all outcomes in the first $k$ rounds being $0$s is in proportion to: 
\begin{align*}
    \sum_i \alpha_i p_i^{k/2} \begin{cases}
        \ket{v^0_i}  \text{ if  $k$ is even}, \\
        \ket{w^0_i}  \text{ if  $k$ is odd}. 
    \end{cases}
\end{align*}

The probability of all outcomes being $0$s is $\sum_i |\alpha_i|^2 p_i^k$. 
\end{lemma}
The proof follows the proof of Claim 6.3 in \cite{zhandry2020schrodinger}. We reprove this claim for completeness.
\begin{proof}
This lemma holds for $k=0$, when no measurement is applied. This is the state is
\begin{align*}
    \sum_i \alpha_i \ket{v^0_i} = \sum_i \alpha_i \ket{\mathbbm{1}_{\mc{R}}}_\rR \ket{\phi_i}_\rA = \ket{\mathbbm{1}_{\mc{R}}}_\rR \ket{\sigma_H, 0^L}_\rA.
\end{align*}
We now prove by induction. Assume the lemma holds up to some even $k$. We prove it holds for odd $k+1$. 

The leftover state after the first $k$ rounds is $c \sum_i \alpha_i p_i^{k/2} \ket{v_i^0}$ for some normalization $c$. Note that $\ket{v^0_i} = \sqrt{p_i} \ket{w_i^0} + \sqrt{1 - p_i} \ket{w^1_i}$. The state can be rewritten as 
\begin{align*}
    c \sum_i \alpha_i p_i^{k/2} \left(\sqrt{p_i} \ket{w_i^0} + \sqrt{1 - p_i} \ket{w^1_i}\right).
\end{align*}

In the $(k+1)$-th round, the challenger measures the state under $\CP^H$. Note that $\CP^H_0 \ket{w_i^0} = \ket{w_i^0}$ and $\CP^H_0 \ket{w_i^1} = 0$. Thus, conditioned on the $(k+1)$-th outcome being $0$, the state is in proportion to 
$\sum_i \alpha_i p_i^{(k+1)/2}  \ket{w_i^0}$.
We complete the induction for $k$ being even. 

For odd $k$, the analysis is almost identical, by observing $\ket{w_i^0} = \sqrt{p_i} \ket{v_i^0} + \sqrt{1-p_i} \ket{v_i^1}$ and also following from the fact that $\IsUniform^0 \ket{v^0_i} = \ket{v^0_i}$ and $\IsUniform^1 \ket{v^0_i} = 0$.

\medskip

Finally, the probability can be bounded by looking at the un-normalized states above. 
\end{proof}

\Cref{thm:misprob} follows from summing over all functions $H$ and \Cref{lem:amprogress}.

\fi

\subsection{Advantages of Uniform Algorithms in Alternating Measurement Games}

In this section, we relate success probabilities of {\bf uniform} quantum algorithms in alternating measurements with probabilities in the corresponding bit-fixing model. We will show the following theorem: 
\begin{theorem}\label{thm:mistobf}
Let $G$ be a game in the QROM and $\As$ be any {\bf uniform} quantum algorithm for $G$ making $T$ oracle queries.
Let $\nu(P, T)$ be the security of $G$ in the $P$-BF-QROM. 
For every $k > 0$, every $P \geq k \, (T + T_\samp + T_\ver)$, 
\begin{align*}
    \epsilon_\As^{\otimes k} \leq \nu(P, T)^k. 
\end{align*}

Recall that $T_\samp, T_\ver$ are the numbers of queries made by $\samp$ and $\ver$, respectively. 
\end{theorem}

\medskip

To bound $\epsilon_\As^{\otimes k}$ for any uniform quantum algorithm, it is sufficient to bound the following conditional probability: $\epsilon_\As^{(t)}$ for $t = 1, \cdots, k$.
\begin{definition}[Conditional Probability for the $t$-th Outcome]\label{def:conditionalprob}
$\epsilon_\As^{(t)}$ is the conditional probability $\Pr[b_t = 0 \,|\, {\bf b}_{<t} = \mb{0}]$, where ${\bf b}_{<t}$ and $b_t$ are the first $t$ outcomes produced by the game $G^{\otimes k}$ with $\As$, when $H$ is picked uniformly at random.  
\end{definition}

Next, we characterize the conditional probability in terms of eigenvalues $\{p_{H,j}\}_j$ and amplitudes under the corresponding eigenbasis $\{\ket{\phi_{H,j}}\}_j$. 
\begin{lemma}
Let $G$ be a game and $\As = (\{U_\inp\}_{\inp})$ be any {\bf uniform} quantum algorithm for $G$. Let $P_H$ be the corresponding POVMs for function $H$. Let $\{\ket{\phi_{H, j}}\}_j$ be the set of eigenbasis for $P_H$ with eigenvalues $\{p_{H, j}\}_j$. 

For each $H$, write the starting state $\ket{0^S}\ket{0^L}$ as $\sum_{i} \alpha_{H, i} \ket{\phi_{H, i}}$. Let $\epsilon^{(t)}_\As$ for $1 \leq t \leq k$ be the conditional probability defined in \Cref{def:conditionalprob}. Then
\begin{align*}
    \epsilon^{(t)}_\As = \frac{ \sum_{H, i} |\alpha_{H,i}|^2 \cdot p^{t}_{H, i}}{ \sum_{H, i} |\alpha_{H,i}|^2 \cdot p^{t-1}_{H, i}}. 
\end{align*}
\end{lemma}
\begin{proof}
By definition, $\epsilon^{(t)}_\As = \Pr[b_t = 0 \,|\, \mb{b}_{< t} = \mb{0}] = \Pr[\mb{b}_t = \mb{0}] / \Pr[\mb{b}_{t-1} = \mb{0}]$. Since $\Pr[\mb{b}_k = \mb{0}] = \sum_{H, i} |\alpha_{H,i}|^2 \cdot p^{k}_{H, i}$, we conclude the lemma. 
\end{proof}

In order to bound $\epsilon_\As^{\otimes k}$, it is enough to bound $\epsilon_\As^{(t)}$ for every $1 \leq t \leq k$ and $\epsilon_\As^{\otimes k} = \prod_{1 \leq t \leq k} \epsilon^{(t)}_{\As}$. Indeed, with \Cref{lem:prob_monotone}, we have the following straightforward corollary. 
\begin{corollary}\label{cor:monotone_mis_bound}
For every game $G$ and {\bf uniform} quantum algorithm $\As$, $\{\epsilon^{(t)}\}_{t \geq 1}$ is monotonically non-decreasing.  Therefore, $\epsilon_\As^{\otimes k} \leq \left(\epsilon_\As^{(k^*)}\right)^k$ for any $k^* \geq k$. 
In particular, $\epsilon_\As^{\otimes k} \leq \left(\epsilon_\As^{k}\right)^k$. 
\end{corollary}
\begin{proof}
The proof is direct by setting $\{c_i\}$, $\{p_i\}$ in the statement of \Cref{lem:prob_monotone} as $\left\{|\alpha_{H, i}|^2 \cdot p_{H, i}^{t} / N^M\right\}$ and $\{p_{H, i}\}$.
\end{proof}

Finally, we show a connection between $\epsilon_\As^{(k)}$ and $\nu(P, T)$ of the game $G$ in the $P$-BF-QROM for $P \geq k \, (T + T_\samp + T_\ver)$.
\begin{lemma}\label{lem:mis_to_bf}
For every game $G$ and {\bf uniform} quantum $T$-query algorithm $\As$, every {\bf odd} $k > 0$, every $P \geq (k - 1) \, (T+T_\samp+T_\ver)$, 
\begin{align*}
    \epsilon_\As^{k} \leq \nu(P, T).
\end{align*}

As a direct corollary by the monotonicity of $\epsilon_\As^{(t)}$, for {\bf even} $k > 0$, every $P \geq k (T+T_\samp+T_\ver)$, 
\begin{align*}
    \epsilon_\As^{k} \leq \epsilon_\As^{(k+1)} \leq \nu(P, T).
\end{align*}
\end{lemma}

Together with \Cref{cor:monotone_mis_bound}, we conclude the main theorem (\Cref{thm:mistobf}) in this subsection.

\begin{proof}[Proof for \Cref{lem:mis_to_bf}]

We only need to prove the lemma for odd $k$ (or even $(k-1)$).

Recall in \Cref{def:pbfqrom}, we need to specify a $P$-query quantum algorithm $f$ and a $T$-query algorithm $\Bs$ to describe an algorithm in the $P$-BF-QROM.
The game is executed if and only if $f^H$ outputs $0$. We define $f, \Bs$ as follows (\Cref{fig:PBF-algo}). 

\begin{figure}[!hbt]
    \centering
    \begin{gamespec}
    \begin{description}
    \item $P$-query quantum algorithm $f$: 
    \begin{itemize}
        \item Initialize $\ket{\mathbbm{1}_{\mc{R}}}_\rR \ket{0^S,0^L}_{\rA}$.
        \item Run the alternating measurement game for $(k-1)$-rounds (\Cref{def:misviaam}).
        Let $\tau$ be the leftover state. 
        \item Let a boolean variable $b = 0$ if and only if all outcomes in $(k-1)$-rounds are $0$s.
        \item Output $b$ and $\tau_{\rR\rA}$. 
    \end{itemize}
    \item $T$-query online algorithm $\Bs$:
    \begin{itemize}
        \item Take $\tau_{\rR\rA}$ as input.
        \item On an online challenge $\ch \gets \samp^H(r)$, it runs $\As$ on internal state $\tau[\rA]$ and outputs the answer produced by $\As$. 
    \end{itemize}
    \end{description}
    \end{gamespec}
    \caption{Turn $\As$ into an algorithm in the $P$-BF-QROM.}
    \label{fig:PBF-algo}
\end{figure}

First, we show that $(f, \Bs)$ is a $(P, T)$ algorithm in the $P$-BR-QROM. It is easy to see that $\Bs$ makes at most $T$ queries as $\As$ makes at most that many queries. The number of queries made by $f$ is equal to that made in the alternating measurement game: 
\begin{itemize}
    \item In odd rounds, one needs to apply $\CP^H$, which takes $2 (T+T_\samp) + T_\ver$ queries; here $2 (T+ T_\samp)$ is for both $U^H_{\samp^H(r)}$ and its inverse $\left(U^H_{\samp^H(r)}\right)^\dagger$ and $T_\ver$ is for applying the projection $V^H_r$ (recall the definitions in \Cref{sec:notations}). 
    \item In even rounds, no queries are needed. 
\end{itemize}
Thus, when $(k-1)$ is even, the total number of queries is at most $(k-1) (T+T_\samp + T_\ver)$. 

\medskip

Next we prove that $(f, \Bs)$ succeeds with probability $\epsilon_\As^{(k)}$. Thus by the definition of $\nu(P, T)$, $\epsilon_\As^{(k)}$ is at most $\nu(P, T)$, concluding the lemma. 

For a fixed hash function $H$ and even $(k-1)$ (or equivalently, odd $k$), conditioned on $f^H$ outputting $0$, the leftover state $\tau_{\rR\rA}$ is (by \Cref{lem:amprogress}):
\begin{align*}
    \tau_{\rR\rA} \propto \sum_i \alpha_{i} p_{i}^{(k-1)/2} \ket{v_i^0}_{\rR\rA} = \ket{\mathbbm{1}_{\mc{R}}}_{\rR} \otimes \sum_{i} \alpha_i p_i^{(k-1)/2} \ket{\phi_{ i}}_{\rA}.
\end{align*}
Here we ignore $H$ for subscripts or superscripts. 

Therefore, $\tau[\rA] = c \sum_{i} \alpha_i p_i^{(k-1)/2} \ket{\phi_{i}}_{\rA}$ where $c$ is a normalization factor such that $1/c^2 = \sum_i |\alpha_i|^2 p_{i}^{k-1}$. The winning probability of $\Bs$ for this fixed $H$ is 
\begin{align*}
    \mathbb{E}_r\left[\left|V^H_r U^H_{\samp^H(r)} \tau[\rA]\right|^2\right]& = c^2 \sum_i |\alpha_i|^2 p_i^{(k-1)} \langle \phi_i | P_H | \phi_i\rangle \\
    &= c^2 \sum_i |\alpha_i|^2 p_i^k,
\end{align*}

By taking the weighted sum of the winning probability for each $H$, the winning probability of $\Bs$ is
\begin{align*}
    \frac{\sum_{H, i} |\alpha_{H, i}|^2 p_{H,i}^{k}}{\sum_{H, i} |\alpha_{H, i}|^2 p_{H,i}^{k-1}} = \epsilon_{\As}^{(k)}. 
\end{align*}
Finally, since $G$ is $\nu(P, T)$ secure in the $P$-BF-QROM, $\epsilon_\As^{(k)} \leq \nu(P, T)$ for every $T$ query quantum algorithm $\As$ and $P \geq (k-1) (T + T_\samp + T_\ver)$.
\end{proof}

Lastly, we prove \Cref{thm:mistobf}. 
\begin{proof}[Proof for \Cref{thm:mistobf}]
    It follows easily by combining \Cref{cor:monotone_mis_bound} and \Cref{lem:mis_to_bf}. 
\end{proof}

\subsection{Proof of Main Theorem}

In this section, we prove our main theorem, \Cref{thm:bftononuniform}.

We start by proving the first part of the theorem.
\begin{proof}[Proof for the first part]  
Let $G$ be any game. For any $S, T$, let $k = S$ and $P = k (T + T_\samp + T_\ver) = S (T + T_\samp+T_\ver)$. $G$ is $\nu(P, T)$ secure in the $P$-BF-QROM. 

By \Cref{thm:mistobf}, for any uniform $T$-query quantum algorithm and $k = S$, its winning probability in the alternating measurement game $G^{\otimes k}$ is at most $\nu(P,T)^k$. 

Therefore, for any $(S, T)$ non-uniform quantum algorithm $\As$, its success probability $\epsilon_\As^{\otimes k}$ is at most $2^S \nu(P, T)^k = (2 \nu(P, T))^S$. This is because for any non-uniform algorithm of winning probability $p$ with advice being an $S$-bit advice $\ket{\sigma_H}$, we can turn it into a uniform quantum algorithm with winning probability at least $2^{-S} p$ as follows (\cite{aaronson2004limitations}):
\begin{description}
    \item As the uniform algorithm does not know $\ket{\sigma_H}$, it samples an $S$-qubit maximally mixed state and runs the non-uniform algorithm on the maximally mixed state.
\end{description}
Since an $S$-qubit maximally mixed state can be written as $1/2^S \ket{\sigma_H}\bra{\sigma_H} + (1-1/2^S) \sigma'$, the uniform algorithm has success probability at least $p/2^S$.

Finally, due to \Cref{lem:non-uniform_to_mis}, any non-uniform algorithm $\As$ is at most $2 \nu(P, T)$ secure in $G$ for $P = S (T + T_\samp + T_\ver)$. 
\end{proof}

The proof for the second part is similar but more laborious. Since we are dealing with decision games, we need to carefully deal with the factor $2^{-S}$ in the previous proof.
\begin{proof}[Proof for the second part]
The theorem trivially holds when $\gamma \geq 1$. We prove it for $\gamma \in (0,1]$.

Let $G$ be a decision game. For any $P, T$, $G$ is $\nu(P, T)$ secure in the $P$-BF-QROM. 

Similarly by \Cref{thm:mistobf}, for any uniform $T$-query quantum algorithm and $k$, its security in the alternating measurement game $G^{\otimes k}$ is at most $\nu(P,T)^k$ where $P = k (T + T_\samp + T_\ver)$. Thus, for any $(S, T)$ non-uniform quantum algorithm $\As$, $\epsilon_\As^{\otimes k}$ is at most $2^{S} \nu(P,T)^k$. 

Since for any $\gamma \in (0,1]$,  $2 \leq (1+\gamma)^{1/\gamma}$. By setting $k = S/\gamma$, we have:
\begin{align*}
    \epsilon_\As^{\otimes k} \leq 2^S \nu(P, T)^k \leq \left( (1 + \gamma) \nu(P, T) \right)^k\leq \left( \frac{1}{2} + \nu'(P, T) + \gamma \right)^k. 
\end{align*}
The last inequality follows the union bound and $\nu(P, T) = 1/2 + \nu'(P,T)$. 

Since the above inequality holds for all $\gamma \in (0,1]$, we conclude the second part of our theorem, following \Cref{lem:non-uniform_to_mis}.

\end{proof}

\ifllncs 
\else

\section{Applications}
\label{sec:application}

We show several applications of our main theorem (\Cref{thm:bftononuniform}) in this section. We first apply our theorem to OWF and PRG games and achieve improved lower bounds for both games. The former ones are publicly verifiable, and the latter games are decision games and thus not publicly verifiable. The applications for both types of games show our main theorem is general and achieve pretty good bounds for almost all kinds of security games in the QROM against quantum/classical advice, as long as we can analyze their security in the $P$-BF-QROM. 

Finally, we show that ``salting defeats preprocessing'' in the QROM, which extends the classical theorem by Coretti et al. \cite{coretti2018random} and improved the result by Guo et al. \cite{chung2020tight}.

\paragraph{OWF.}
Recall the definition of $G_{\sf OWF}$ in \Cref{ex:OWFandPRG}. It is shown that $G_{\sf OWF}$ has the following security in the in the $P$-BF-QROM, $\nu(P, T) = O\left((P + T^2) / \min\{N,M\}\right)$, where $N$ and $M$ are the sizes of the domain and range of the random oracle, by Lemma 1.5 in \cite{chung2020tight}.  

By our main theorem \Cref{thm:bftononuniform}, we have the following theorem. 
\begin{theorem}\label{prop:owf}
$G_{\sf OWF}$ has security $\delta(S, T) = O\left(\frac{S T + T^2}{\min\{N, M\}}\right)$ against $(S, T)$ non-uniform quantum adversaries, even with quantum advice. 
\end{theorem}
The above theorem improves the bound for quantum advice, which was shown to be $\tilde{O}\left( \frac{ST + T^2}{\min\{N,M\}} \right)^{1/3}$ in \cite{chung2020tight}.

\paragraph{PRG.} Recall $G_{\sf PRG}$ is defined in \Cref{ex:OWFandPRG}. $G_{\sf PRG}$ has security $\nu(P, T) = 1/2+O\left(\frac{P + T^2}{N}\right)^{1/2}$ where $N$ is the size of the domain, by Lemma 1.6 in \cite{chung2020tight}.  
Again by our main theorem \Cref{thm:bftononuniform}, we have the following theorem. 
\begin{theorem}\label{prop:prg}

$G_{\sf PRG}$ has security $\delta(S, T) = 1/2 +  O\left(\frac{T^2}{N}\right)^{1/2} + O\left(\frac{S T}{N}\right)^{1/3}$ against $(S, T)$ non-uniform quantum adversaries, even with quantum advice.
\end{theorem}

This improves the previous result on $G_{\sf PRG}$ with quantum advice \cite{chung2020tight}, which was $1/2+\tilde{O}\left( \frac{S^5 T + S^4 T^2}{N} \right)^{1/19}$. 

\subsection{Salting Defeats Quantum Advice}

We start by defining the cryptographic mechanism called ``salting''.
\begin{definition}[Salted Games in the QROM]\label{def:salted_game}
Let $G$ be a game in the QROM as defined in \Cref{def:game}, with respect to a random oracle $H:[N] \to [M]$. It consists of two deterministic algorithms $\samp^H$ and $\ver^H$ and both algorithms make $T_\samp$ (or $T_\ver$) queries, respectively. 

A salted game $G_S$ with salt space $[K]$ is defined as the following: $G_S$ consists of two deterministic algorithms $\samp_S$ and $\ver_S$: 
\begin{itemize}
    \item $\samp^H_S$: on input $s, r$, it returns $(s, \samp^{H_s}(r))$. Here $H_s$ denotes oracle access to the oracle $H(s,\cdot)$.
    \item $\ver^H_S$: on input $s, r, \ans$, it returns $\ver^{H_s}(r, \ans)$.
\end{itemize}

In other words, for a fixed $H:[K] \times [N] \to [M]$ and a quantum algorithm $\As$, the game $G^H_{S, \As}$ is executed as follows:
\begin{itemize}
    \item A challenger $\Cs$ samples a uniformly random salt $s \gets [K]$ and $\ch \gets \samp^{H_s}(r)$ using uniformly random coins $r$. 
    \item A (uniform or non-uniform) quantum algorithm $\As$ has oracle access to $H$, takes $(s, \ch)$ as input and outputs $\ans$. 
    \item $b \gets \ver^{H_s}(r, \ans)$ is the outcome of the game. 
\end{itemize}
\end{definition}

\begin{lemma}[Salted Games in the $P$-BF-QROM, Lemma 7.2 in \cite{chung2020tight}] \label{lem:saltinpbfqrom}

Let $G$ be a game in the QROM, with security $\nu(T)$ against $T$-query quantum adversaries. Then for any $P$,
\begin{itemize}
    \item $G$ has security $\nu(P, T) \leq 2 \nu(T) + O(P / K)$ in the $P$-BF-QROM; 
    \item $G$ has security $\nu(P, T) \leq \nu(T) + O(\sqrt{P/K})$ in the $P$-BF-QROM. 
\end{itemize}
\end{lemma}
The second bullet point is better than the first one, when $G$ is a decision game.

\begin{proof}
The proof is subsumed by  the proof for Lemma 7.2 \cite{chung2020tight}. Although Lemma 7.2 shows the multi-instance security of $G_S$, its $P$-BF-QROM security is an intermediate step. 
\end{proof}

Combining with \Cref{thm:bftononuniform}, we have the following results about salting in the QROM.
\begin{theorem}\label{prop:salt}
For any game $G$ (as defined in \Cref{def:game}) in the QROM, let $\nu(T)$ be its security in the QROM. 
Let $G_S$ be the salted game with salt space $[K]$. Then $G_S$ has security $\delta(S, T)$ against $(S, T)$ non-uniform quantum adversaries with quantum advice,
\begin{itemize}
    \item $\delta(S, T) \leq 4 \nu(T) + O(S (T + T_\samp + T_\ver)/K)$;
    \item If $G_S$ is a decision game, then $\delta(S, T) \leq \nu(T) + O({S (T + T_\samp + T_\ver) /K})^{1/3}$. 
\end{itemize}
\end{theorem}
\begin{proof}
We only show the second bullet point. The first one is similar and more straightforward. 

By \Cref{thm:bftononuniform}, $\delta(S, T) \leq \min_{\gamma > 0} \left\{ \gamma + \nu(P/\gamma, T) \right\}$ where $P = S(T + T_\ver + T_\samp)$. Since $\nu(P/\gamma, T) \leq \nu(T) + O(\sqrt{P/(K \gamma)})$ by \Cref{lem:saltinpbfqrom}, $\delta(S, T)$ takes its minimum when $\gamma = O(P/K)^{1/3}$. Our second result follows. 
\end{proof}

\fi

\section{Advantages of Quantum Advice in the QROM}
\label{sec:separation}

This section demonstrates a game in which non-uniform quantum algorithms with quantum advice have an exponential advantage over those with classical advice for some parameter regime $S, T$. Although the advantage only applies to some $S, T$ ranges \footnote{Specifically, we require $T = 0$, i.e., no online query.}, we believe it is the first step toward understanding a game in which quantum advice has an exponential advantage over classical advice for a wider range of $S, T$. 

The game is based on the recent work by Yamakawa and Zhandry \cite{yamakawa2022verifiable}. We start by explaining and recalling the basic ideas in their work. 
\begin{definition}[ \cite{yamakawa2022verifiable}, YZ Functions]
Let $n$ be a positive integer, $\Sigma$ be an exponentially (in $n$) sized alphabet and $C \subseteq \Sigma^n$ be an error correcting code over $\Sigma$.  
Let $H : [n] \times \Sigma \to \{0,1\}$ be a random oracle. The following function is called a YZ function with respect to $C$ and $\Sigma$: 
\begin{align*}
    f^H_C : C \to \{0,1\}^n & \\
    f^H_C(c_1, c_2, \cdots, c_n) &= H(1, c_1) || H(2, c_2) || \cdots || H(n, c_n)
\end{align*}
\end{definition}

We will consider the following game, which we call $G_{\sf YZ}$. The game is to invert a uniformly random image with respect to the YZ function. More formally,  
\begin{definition}[Inverting YZ Functions]
The game $G_{\sf YZ}$ is specified by two classical algorithms:
\begin{itemize}
    \item $\samp^H(r)$: it samples a uniformly random image $y = r \in \{0,1\}^n$;
    \item $\ver^H(r, \ans)$: it checks whether $\ans$ is a code in $C$ and $f^H_C(\ans) = r$. 
\end{itemize}
The queries made by each algorithm satisfy $T_\samp = 0$ and $T_\ver = n$. 
\end{definition}
Their idea is that, if we want to find a pre-image in $\Sigma^n$ of any $y \in \{0,1\}^n$, it is easy: simply inverting each $H(i, y_i)$. Nevertheless, to find a pre-image in $C$, this entry-by-entry brute-force no longer works. 
In their work, Yamakawa and Zhandry show that for some appropriate $C$, the above function is classically one-way and quantumly easy to invert. 
\begin{theorem}[Theorem 6.1, Lemma 6.3 and 6.9 in \cite{yamakawa2022verifiable}] \label{thm:YZ}
There exists some appropriate $C$, such that 
\begin{itemize}
    \item The game $G_{\sf YZ}$ has security $2^{-\Omega(n)}$ against $2^{n^c}$-query classical adversaries for some constant $0 < c < 1$; 
    \item There is a $\tilde{O}(n)$-query quantum algorithm that wins the game $G_{\sf YZ}$ with probability $1-\negl(n)$. Here $\tilde{O}$ hides a polylog factor. 
\end{itemize}
\end{theorem}

Moreover, we observe that the quantum algorithm makes non-adaptive queries and the queries are independent of the challenge. Upon a challenge $y$ is received, the quantum algorithm does post-processing on the quantum queries without making further queries \footnote{For more details, please refer to Fig 1. in \cite{yamakawa2022verifiable}}. 

\medskip

We show our separation result below. 
\begin{theorem}[Separation of classical and quantum advice in the QROM]
There exists some appropriate $C$ (the same in \cite{yamakawa2022verifiable}) such that, 
\begin{itemize}
    \item $G_{\sf YZ}$ has security $2^{-\Omega(n)}$ against $(S, T = 0)$ non-uniform adversaries with {\bf classical} advice, for $S = 2^{n^c}/n$ and some constant $0 < c < 1$; 
    \item There is an $(S, T = 0)$ non-uniform adversary with {\bf quantum} advice that achieves success probability $1 - \negl(n)$, for $S = \tilde{O}(n)$. 
\end{itemize}
\end{theorem}

\ifllncs
We refer readers to a detailed proof in the appendix. 
\else 
\begin{proof}
We first show the second bullet point. 
Let the quantum algorithm in \Cref{thm:YZ} be $\Bs$. 
In the non-uniform quantum adversary, quantum advice is the non-adaptive queries made by $\Bs$ and the online stage is the post-processing by $\Bs$. It is straightforward that the non-uniform algorithm achieves the same probability as $\Bs$, which is $1 - \negl(n)$. Since each query has $O(\log n)$ qubits and $\Bs$ makes $\tilde{O}(n)$ queries, the total size of the quantum advice is still $\tilde{O}(n)$.

Next, we show the first bullet point. In the first bullet point of this theorem, we do not distinguish between non-uniform quantum adversaries with classical advice and non-uniform classical adversaries. The reason is that the online algorithm does not make any query, i.e., $T = 0$. These two types of algorithms are equivalent when $T = 0$. 

Thus, we consider success probabilities of non-uniform classical adversaries. By a classical analog of our main theorem \Cref{thm:bftononuniform} (\Cref{thm:classical_bf_to_nonuniform}), we only need to show its success probability in the $P$-BF-ROM (\Cref{def:pbfrom}) where $P = S (T + T_\samp + T_\ver) = S T_\ver = 2^{n^c}$. 

Assume a random oracle is lazily sampled. In other words, an outcome of the random oracle on $x$ is sampled only if the outcome is queried by an algorithm; otherwise, the outcome is marked as ``not sampled''.  
Conditioned on any $P$-query $f$ outputs $0$, the random oracle is only fixed on $P$ positions and the rest of its outputs are still not sampled. The error correcting code $C$ used in \cite{yamakawa2022verifiable} satisfies a property called $(\zeta, \ell, L)$ list recoverability: 
\begin{itemize}
    \item For any subset $S_i \subseteq \Sigma$ such that $|S_i| \leq \ell$ for every $i \in [n]$, we have 
    \begin{align*}
        |{\sf Good}| = \left| \left\{ (x_1, \cdots, x_n) \in C : |\{ i \in [n]: x_i \in S_i\}| \geq (1-\zeta) n \right\}  \right| \leq L.
    \end{align*}
    In other words, the total number of codewords in $C$ with hamming distance to $S_1 \times S_2 \times \cdots \times S_n$ smaller than $\zeta n$ is bounded by $L$. Here hamming distance to $S_1 \times S_2 \times \cdots \times S_n$ is defined as the number of coordinates $i$ whose $x_i$ is not in the corresponding $S_i$.
    
    We call this set of codewords ${\sf Good}$.
    
    \item $P = 2^{n^c} < \ell$, $\zeta = \Omega(1)$ and $L = 2^{n^{c'}}$ for some $0 < c' < 1$. 
\end{itemize}

In $G_{\sf YZ}$, when a challenge $y$ is sampled uniformly at random from $\{0,1\}^n$, there are two cases:
\begin{itemize}
    \item \textbf{Case 1}: there exists a codeword $c$ in ${\sf Good}$, such that $y = f^H_C(c)$. This case happens with probability at most $|{\sf Good}|/2^n \leq L / 2^n$.
    \item \textbf{Case 2}: complement of Case 1. In this case, an adversary wins only if it outputs a codeword that is not in ${\sf Good}$. 
    
    For every codeword $c = (x_1, x_2, \cdots, x_n) \not\in {\sf Good}$, there are at least $\zeta n$ coordinates whose random oracle outputs (i.e., $H(i,x_i)$) have not been sampled yet in the lazily sampled random oracle. For any $c \not \in {\sf Good}$, $\Pr[f^H_C(c) = y] \leq 2^{-\zeta n}$. Therefore, regardless of the algorithm's output, the success probability is at most $2^{-\zeta n}$. 
\end{itemize}

The overall winning probability is bounded by $L/2^n + 2^{-\zeta n} = 2^{-\Omega(n)}$. We conclude the first bullet point of the theorem. 
\end{proof}
\fi

\ifllncs
  \bibliographystyle{alpha}
  \bibliography{bib}
\else
  \printbibliography
\fi

\appendix

\ifllncs 

\else 
\fi 

\ifllncs 
\input{appendix_helpfullemma}

\section{\texorpdfstring{Characterization of Alternating Measurements and Proof of \Cref{thm:misprob}}{Characterization of Alternating Measurements}}

\fi

\section{Classical Version of Our Main Theorem}
\label{sec:classcai_analog}
The following theorem is a classical version of our main theorem (\Cref{thm:bftononuniform}), improved from Theorem 1 in \cite{guo2021unifying}. 

\begin{theorem}\label{thm:classical_bf_to_nonuniform}
Let $G$ be any game with $T_\samp, T_\ver$ being the number of queries made by $\samp$ and $\ver$. 
For any $S, T$, let $P = S (T + T_\ver + T_\samp)$. 

If $G$ has security $\nu(P, T)$ in the $P$-BF-ROM, then it has security $\delta(S, T) \leq 2 \cdot \nu(P, T)$ against $(S, T)$ non-uniform classical algorithms with classical advice. 
\end{theorem}

In Theorem 1 in \cite{guo2021unifying}, $P =  (S + \log \gamma^{-1}) (T + T_\ver + T_\samp)$ and there is an extra additive term $\gamma$ for $\delta(S, T)$. 
\begin{theorem}[Theorem 1 in \cite{guo2021unifying}]
Let $G$ be any game with $T_\samp, T_\ver$ being the number of queries made by $\samp$ and $\ver$. 
For any $S, T, \gamma > 0$, let $P =  (S + \log \gamma^{-1}) (T + T_\ver + T_\samp)$. 

If $G$ has security $\nu(P, T)$ in the $P$-BF-ROM, then it has security $\delta(S, T) \leq 2 \cdot \nu(P, T) + \gamma$ against $(S, T)$ non-uniform classical algorithms with classical advice. 
\end{theorem}

\ifllncs 
\input{appendix_seperation_proof}
\fi 

\end{document}